\def\full{1} 
\def\final{1} 

\ifnum\final=0
\newcommand{\mynote}[3]{\todo[color=#2]{\tiny { {\sc #1}: {\sf #3}}}}
\else
\newcommand{\mynote}[3]{}
\fi

\newcommand{\asnote}[1]{\mynote{Adam}{yellow}{#1}}
\newcommand{\rbnote}[1]{\mynote{Raef}{green}{#1}}

\ifnum\full=1
\documentclass[11pt]{article}
\else
\documentclass{sig-alternate-2013}
\newfont{\mycrnotice}{ptmr8t at 7pt}
\newfont{\myconfname}{ptmri8t at 7pt}

\permission{Permission to make digital or hard copies of all or part of this work for personal or classroom use is granted without fee provided that copies are not made or distributed for profit or commercial advantage and that copies bear this notice and the full citation on the first page. Copyrights for components of this work owned by others than ACM must be honored. Abstracting with credit is permitted. To copy otherwise, or republish, to post on servers or to redistribute to lists, requires prior specific permission and/or a fee. Request permissions from permissions@acm.org.}
\conferenceinfo{STOC'15,}{June 14--17, 2015, Portland, Oregon, USA.\\
{\mycrnotice{Copyright is held by the owner/author(s). Publication rights licensed to ACM.}}}
\copyrightetc{ACM \the\acmcopyr}
\crdata{978-1-4503-3536-2/15/06\ ...\$15.00.\\
DOI: http://dx.doi.org/10.1145/2746539.2746632}

\usepackage{times}
\usepackage{graphicx,color}
\usepackage{array,float}
\usepackage[usenames,dvipsnames]{xcolor}
\usepackage{amstext,amssymb,amsmath}
\usepackage{verbatim}
\usepackage{bm}
\usepackage{paralist}
\usepackage{ulem}
\usepackage[numbers]{natbib}
\usepackage{todonotes}
\usepackage[noend]{algorithmic}
\usepackage{algorithm}

\newtheorem{lem}{Lemma}[section]
\newtheorem{thm}[lem]{Theorem}

\newtheorem{defn}[lem]{Definition}

\newtheorem{claim}[lem]{Claim}
\renewcommand{\algorithmicrequire}{\textbf{Input:}}
\renewcommand{\algorithmicensure}{\textbf{Output:}}

\renewcommand{\cite}{\citep}
\newcommand{\thet}{\theta}
\newcommand{\linfty}[1]{\|#1\|_\infty}
\newcommand{\ip}[2]{\langle #1,#2\rangle}
\newcommand{\nptheta}{\hat\theta}
\newcommand{\symdiff}{\Delta}
\newcommand{\privtheta}{\theta^{priv}}
\renewcommand{\paragraph}[1]{\vspace{3pt}\noindent\textbf{#1}}
\newcommand{\scrX}{\ensuremath{\mathcal{X}}}
\newcommand{\scrY}{\ensuremath{\mathcal{Y}}}
\newcommand{\scrZ}{\ensuremath{\mathcal{Z}}}		
\newcommand{\eL}{\mathcal{L}}
\newcommand{\beL}{\bar{\mathcal{L}}}
\newcommand{\bad}{\sf Bad}
\newcommand{\good}{\sf Good}
\newcommand{\hash}{\sf Hash}
\newcommand{\rA}{\ensuremath{\rightarrow}}
\newcommand{\rrA}{\ensuremath{\longrightarrow}}
\newcommand{\qB}{\ensuremath{\mathbf{q}}}
\newcommand{\XB}{\ensuremath{\mathbf{X}}}
\newcommand{\zeroB}{\ensuremath{\mathbf{0}}}
\newcommand{\sm}{\mbox{\textendash}}
\newcommand{\ltwo}[1]{\|#1\|_2}
\newcommand{\lone}[1]{\|#1\|_1}
\newcommand{\linf}[1]{\|#1\|_{\infty}}
\newcommand{\eps}{\epsilon}
\newcommand{\A}{\mathcal{A}}
\newcommand{\cost}{\mathsf{cost}}
\newcommand{\stat}{\mathsf{stat}}
\newcommand{\D}{\mathcal{D}}
\newcommand{\J}{J}
\newcommand{\G}{\mathcal{G}}
\newcommand{\hc}{\mathcal{H}}
\newcommand{\bx}{\mathbf{x}}
\newcommand{\by}{\mathbf{y}}
\newcommand{\bz}{\mathbf{z}}
\newcommand{\hbz}{\hat{\mathbf{z}}}
\newcommand{\bw}{\mathbf{w}}
\newcommand{\bolda}{\mathbf{a}}
\newcommand{\boldb}{\mathbf{b}}
\newcommand{\indx}{\mathsf{index}}
\newcommand{\bit}{\mathsf{bit}}
\newcommand{\bbz}{\bar{\mathbf{z}}}
\newcommand{\bM}{\bar{M}}
\newcommand{\zcp}{\mathcal{Z}_{+}}
\newcommand{\zcn}{\mathcal{Z}_{-}}
\newcommand{\hypsc}{\{-\frac{1}{\epsilon},~\frac{1}{\epsilon}\}^m}
\newcommand{\hypc}{\{-\frac{1}{\sqrt{m}},~\frac{1}{\sqrt{m}}\}^m}
\newcommand{\Pc}{\mathcal{P}}
\newcommand{\hP}{\hat{\mathcal{P}}}
\newcommand{\T}{\mathcal{T}}
\newcommand{\risk}{{\sf R}}
\newcommand{\vol}{{\sf Vol}}
\newcommand{\ind}{{\mathbf{1}}}
\newcommand{\mineig}{\mu}
\newcommand{\I}{\mathbb{I}}
\newcommand{\Ico}{\mathcal{I}_{\sf out}}
\newcommand{\Ici}{\mathcal{I}_{\sf in}}
\newcommand{\E}{\mathbb{E}}
\newcommand{\Sc}{\mathcal{S}}
\newcommand{\Ec}{\mathcal{E}}
\newcommand{\V}{\mathcal{V}}
\newcommand{\F}{\mathcal{F}}
\newcommand{\samp}{\mathsf{Samp}}
\newcommand{\empL}{\mathcal{L}}
\newcommand{\hatw}{\hat{w}}
\newcommand{\dist}{{\sf Dist}_{\infty}}
\newcommand{\hdist}{{\sf Dist}}
\newcommand{\htheta}{\widetilde\theta}
\newcommand{\ptheta}{\theta^\perp}
\newcommand{\dagw}{w^\dagger}
\newcommand{\tildew}{\tilde{w}}
\newcommand{\tildeF}{\tilde{F}}
\newcommand{\tildef}{\tilde{f}}
\newcommand{\re}{\mathbb{R}}
\newcommand{\B}{\mathbb{B}}
\newcommand{\Bc}{\mathcal{B}}
\newcommand{\enc}{{\sf Enc}}
\newcommand{\dec}{{\sf Dec}}
\newcommand{\He}{{\sf Heavhit}}
\newcommand{\coll}{{\sf Coll}}
\newcommand{\splex}{\mathsf{simplex}}
\newcommand{\Q}{\mathcal{Q}}
\renewcommand{\S}{\mathbb{S}}
\newcommand{\teps}{\tilde{\epsilon}}
\newcommand{\hw}{\hat{w}}
\newcommand{\hmu}{\hat{\mu}}
\newcommand{\hmuA}{\hat{\mu}_{\sf A}}
\newcommand{\muA}{\mu_{\sf A}}
\newcommand{\hmuC}{\hat{\mu}_{\C}}
\newcommand{\muC}{\mu_{\C}}
\newcommand{\hmug}{\hat{\mu}_{\good}}
\newcommand{\istr}{i^{\ast}}
\newcommand{\mU}{\mathcal{U}}
\newcommand{\grad}{\bigtriangledown}
\newcommand{\mypar}[1]{\smallskip
\noindent{\bf\em {#1}:}}
\newcommand{\boldpar}[1]{\smallskip
\noindent{\bf{#1}:}}
\newcommand{\etal}{\emph{et al.}}
\newcommand{\ldp}{\bf{LDP}}

\newcommand{\ignore}[1]{}
\newcommand{\tprivJ}{{{\tilde J}^{\text{priv}}}}
\newcommand{\tnonoiseJ}{{{J}^\#}}
\newcommand{\z}{z}
\renewcommand{\b}{b}
\newcommand{\TODO}[1]{{\bf TODO: #1}}
\newcommand{\NOTE}[1]{{\bf NOTE: #1}}
\newcommand{\Ced}{{\Delta_{\epsilon,\delta}}}
\newcommand{\name}{\textsc{GUPT}\xspace}
\newcommand{\nameplain}{GUPT\xspace}
\newcommand{\aname}{a \textsc{GUPT}\xspace}
\newcommand{\Aname}{A \textsc{GUPT}\xspace}
\newcommand{\Weta}{\mathbf{W}^{(\eta)}}
\newcommand{\We}{\mathcal{W}^{(\eta)}}
\newcommand{\W}{\mathcal{W}}
\newcommand{\M}{\mathcal{M}}
\newcommand{\Meta}{\mathcal{M}^{(\eta)}}
\newcommand{\U}{\mathcal{U}}
\newcommand{\R}{\mathcal{R}}
\newcommand{\Z}{\mathcal{Z}}
\newcommand{\be}{\mathbf{e}}
\newcommand{\f}{\mathbf{f}}
\newcommand{\bc}{\mathbf{c}}
\newcommand{\hf}{\hat{f}}
\newcommand{\hbf}{\hat{\mathbf{f}}}
\newcommand{\C}{\mathcal{C}}
\newcommand{\er}{\mathsf{Error}}
\newcommand{\poly}{\mathsf{poly}}
\newcommand{\btr}{\mathsf{Round}}
\newcommand{\genproj}{\mathsf{GenProj}}
\newcommand{\gen}{\mathsf{RndGen}}
\newcommand{\struct}{\mathsf{struct}}
\newcommand{\err}{\textsc{Err}}
\newcommand{\sh}{\mbox{S-Hist}}
\newcommand{\fo}{\mbox{FO}}
\newcommand{\prot}{{\sf PROT}}
\newcommand{\gprot}{{\sf GenPROT}}
\newcommand{\code}{{\sf code}(d, m, \zeta)}
\newcommand{\mmer}{\mathsf{MinMaxError}}
\newcommand{\Lap}{\mathsf{Lap}}
\newcommand{\List}{\textsc{List}}
\newcommand{\pre}{\mathsf{P_{error}}}
\newcommand{\mpre}{\mathsf{P_{min-error}}}

\newcommand{\paren}[1]{{\left({#1}\right)}}

\begin{document}

\title{Local, Private, Efficient Protocols \\ for Succinct Histograms}


\author{Raef Bassily \qquad\qquad Adam Smith\titlenote{Work done while A.S. was on sabbatical at Boston University and Harvard University.}\\
       \affaddr{Department of Computer Science and Engineering}\\
       \affaddr{The Pennsylvania State University}\\
       \email{\{bassily, asmith\}@psu.edu}}

\maketitle

\begin{abstract}
  We give efficient protocols and matching accuracy lower
  bounds for frequency estimation in the local model for differential
  privacy. In this model, individual users randomize their data
  themselves, sending differentially private reports to an untrusted
  server that aggregates them.

  We study protocols that produce a succinct histogram
  representation of the data. A succinct histogram is a list of the most
  frequent items in the data (often called ``heavy hitters'')  along
  with estimates of their frequencies; the frequency of all other
  items is implicitly estimated as 0.

  If there are $n$ users whose items come from a universe of size $d$,
  our protocols run in time polynomial in $n$ and $\log(d)$.  With
  high probability, they estimate the accuracy of every item up to
  error $O(\sqrt{\log(d)/(\eps^2n)})$.
  Moreover, we show that this much error is necessary, regardless of
  computational efficiency, and even for the simple setting where only
  one item appears with significant frequency in the data set.

  Previous protocols (Mishra and Sandler, 2006; Hsu, Khanna
  and Roth, 2012) for this task either ran in time $\Omega(d)$ or had
  much worse error (about $\sqrt[6]{\log(d)/(\eps^2n)}$), and the only
  known lower bound on error was $\Omega(1/\sqrt{n})$. \rbnote{I think the term ``efficient'' here does not fit.}

  We also adapt a result of McGregor et al (2010) to the local
  setting. In a model with public coins, we show that each user need
  only send 1 bit to the server. For all known local
  protocols (including ours), the transformation preserves
  computational efficiency.


\end{abstract}

\section{Introduction}\label{sec:intro}

Consider a software producer that wishes to gather statistics on how
people use its software. If the software handles sensitive
information ---for example, a browser for anonymous web surfing or a
financial management software---users may not want to share their data
with the producer. A producer may not want to collect the raw data
either, lest they be subject to subpoena. How can the producer collect
high-quality \emph{aggregate} information about users while
providing guarantees to its users (and itself!) that it isn't storing
user-specific information?

In the \emph{local model} for private data analysis (also called the
\emph{randomized response} model\footnote{The term ``randomized response''
   may refer either to the model or a specific protocol;
  we use ``local model'' to avoid ambiguity.}), each individual user
randomizes her data herself using a randomizer $Q_i$ to obtain a report
(or ``signal'') $z_i$ which she sends to an
\emph{untrusted} server to be aggregated in to a summary $s$ that can be
used to answer queries about the data (Figure~\ref{fig:local}).
The server may provide public coins visible to all parties, but privacy guarantees depend only
on the randomness of the user's local coins. The local model has been
studied extensively because control of private data remains in users' hands.
\begin{figure}[bht]
\centering    \includegraphics[height=1.25in]{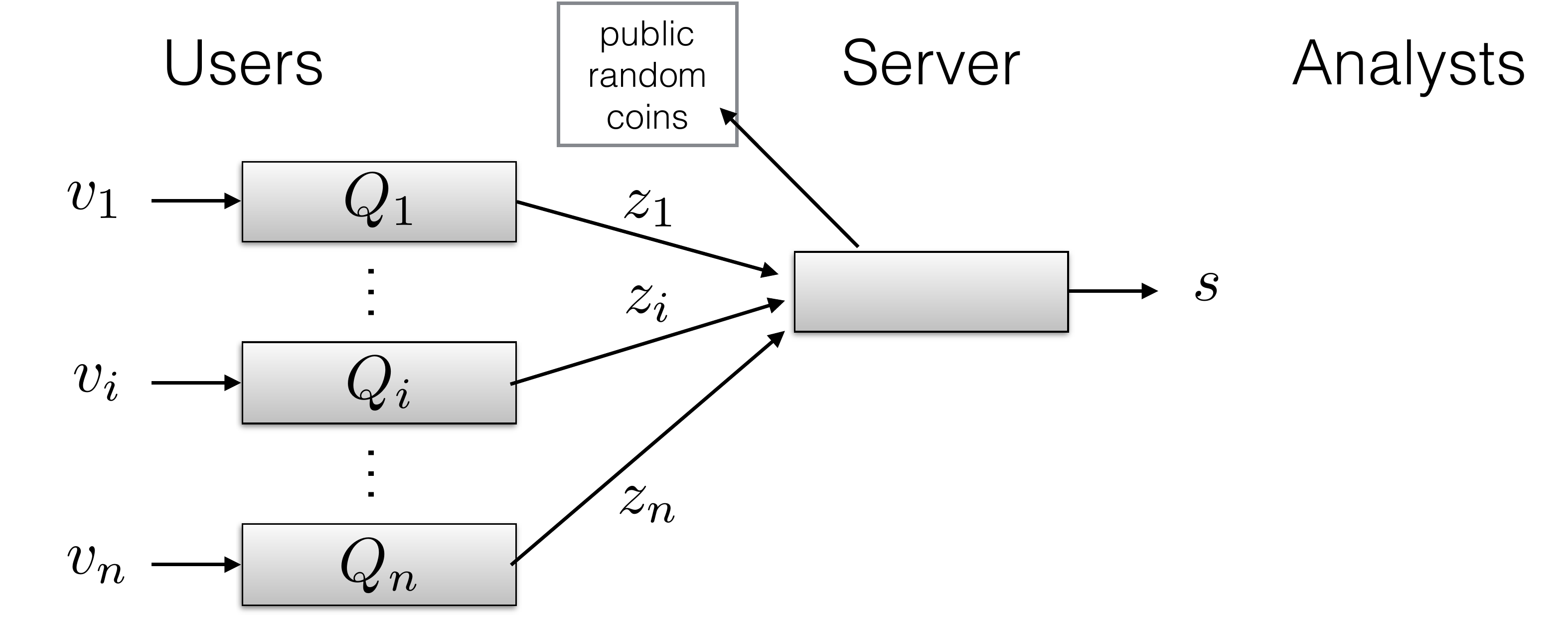}
\caption{The local model for private data analysis.}
\label{fig:local}
\end{figure}

We focus on
protocols that provide \emph{differential privacy} \citep{DMNS06} (or,
equivalently in the local model, \emph{$\gamma$-amplification}
\citep{EGS03} or \emph{FRAPP} \citep{AH05}).
\begin{defn}
We say that an algorithm $\Q:\V\rightarrow\Z$ is $(\eps, \delta)$-local differentially private (or $(\eps, \delta)$-$\ldp$), if for any pair $v, v'\in\V$ and any (measurable) subset $\Sc\subseteq\Z$, we have 
$$\Pr\left[\Q(v)\in\Sc\right]\leq
e^{\eps}\Pr\left[\Q(v')\in\Sc\right]+\delta.$$
The special case with $\delta=0$ is called \emph{pure} $\eps$-$\ldp$.
\label{def:dp}
\end{defn}

We describe new protocols and lower bounds for frequency estimation
and finding heavy hitters in the local privacy model. Local differentially
private protocols for frequency estimation are used in the
Chrome web browser (\citet{Rappor14,Rappor-unknowns15}), and can be used as the basis of
other estimation tasks (see \citet{MS06,DwNi04}). \rbnote{can we note here about the empirical performance of RAPPOR?}
\asnote{Mention multinomial estimation and Warner.}

We also show a generic result for LDP protocols: in the public-coin
setting, each user only needs to send 1 bit to the server.

Suppose that there are $n$ users, and each user $i$ holds a value
$v_i$ in a universe of size $d$ (labeled by integers in
$[d]=\{1,...,d\}$).  We wish to enable an analyst to estimate
frequencies: $f(v)=\frac{1}{n}\#\{i:\ v_i=v\} \, .$ Following
\citet{HKR10}, we look at summaries that provide two types of
functionality:\rbnote{can we add a brief definition for $\#$?}
\begin{itemize}
\item A \textbf{frequency oracle}, denoted $\fo$, is a data structure
  together with an algorithm $A$ that, for any $v\in\V$, allows
  computing an estimate $\hf(v)= A(\fo,v)$ of the frequency $f(v)$.

  The \emph{error} of the oracle $\fo$ is the maximum over items $v$
  of $|\hf(v)-f(v)|$. That is, we measure the $\ell_\infty$ error of
  the histogram estimate implicitly defined by $\hat f$. A protocol
  for generating frequency oracles has error $(\eta,\beta)$ if for all
  data sets, it produces an oracle with error $\eta$ with probability
  at least $1-\beta$.

\item A \textbf{succinct histogram}, denoted $\sh$, is a data
  structure that provides a (short) list of items $\hat v_1,...,\hat
  v_k$, called the \emph{heavy hitters}, together with estimated
  frequencies $(\hf(\hat v_j)\ :\ j\in [k])$.
  The frequencies of the items not in the list are implicitly
  estimated as $\hf(v)=0$. As with the frequency oracle, we measure
  the error of $\sh$ by the $\ell_\infty$ distance between the
  estimated and true frequencies, $\max_{v \in
    [d]}|\hf(v)-f(v)|$.

  If a data structure aims to provide error $\eta$, the
  list need never contain more than $O(1/\eta)$ items (since items with
  estimated frequencies below $\eta$ may be omitted from the list, at the price
  of at most doubling the error).
\end{itemize}

If we ignore computation, these two functionalities are equivalent since
a succinct histogram defines a frequency oracle directly and an analyst with a frequency oracle $\fo$ can query the oracle
on all possible items and retain only those with estimated frequencies
above a threshold $\eta$ (increasing the error by at most
$\eta$). However, when the universe size $d$ is large (for example, if
a user's input is their browser's home page or a financial summary),
succinct histograms are much more useful.

We say a protocol is \emph{efficient} if it has computation time,
communication and storage polynomial in $n$ and $\log (d)$ (the users'
input length).
Prior to this work, efficient protocols for both tasks satisfied only
$(\eps,\delta)$-LDP for $\delta>0$. Efficient protocols for frequency
oracles \citep{MS06,HKR10} were known with worst-case expected error
$O(\sqrt{\frac{\log (d)\log(1/\delta)} {\eps^2 n}})$, while the only
protocols for succinct histograms \citep{HKR10} had much worse error
--- about $\sqrt[6]{\frac{\log (d)\log(1/\delta)} {\eps^2 n}}$. Very
recently, \citet{Rappor-unknowns15} proposed a heuristic construction
for which worst-case bounds are not known.
None
of these protocols matched the best lower bound on accuracy,
$\Omega(1/\sqrt{n})$ \cite{HKR10}.

\subsection{Our Results}

\paragraph{Efficient Local Protocols for Succinct Histograms with
  Optimal Error.} We provide the first
  polynomial time \emph{local} $(\eps,0)$-differentially
  private protocol for succinct histograms that has worst-case error
  $O(\sqrt{\frac{\log (d)} {\eps^2 n}})$. As we show, this error is
  optimal for local protocols (regardless of computation time). Furthermore, in the public
  coin model, each participant sends only 1 bit to the server.
\asnote{Can we also say that computation time is linear in $n$ +
  length of output? We
  might need linear-time decodable codes for that.}

Previous constructions were either inefficient \citep{MS06,HKR10}
(taking time polynomial in $d$ rather than $\log d$), or had much
worse error guarantees%
\footnote{\citet{MS06} state error bounds for a single query to the
  frequency oracle, assuming the query is determined before the
  protocol is executed. Known frequency oracle constructions (both
  previous work and ours) achieve error $O(\sqrt{\log(1/\beta)/n})$ in
  that error model.}%
---at least $\Omega\left(\big(\frac{\log (d)} {\eps^2
  n}\big)^{1/6}\right)$. Furthermore, constructions with communication
sublinear in $d$ satisfied only $(\eps,\delta)$ privacy for
$\delta>0$.

  Our construction consists of two main pieces. Our first protocol
  efficiently recovers a heavy hitter from the input, given a promise
  that the heavy hitter is \emph{unique}: that is, all players either
  have a particular value $v$ (initially unknown to the server) or a
  default value $\bot$. The idea is to have each player send a
  highly noisy version of an error-correcting encoding of their input;
  the server can then recover (the codeword for) $v$ by averaging all
  the received reports and decoding the resulting vector.

  Our full protocol, which works for all inputs, uses ideas from the
  literature on low-space algorithms and compressive
  sensing, e.g., \cite{GGIMS02}. Specifically, using random hashing, we can partition the
  universe of possible items into bins in which there is likely to be only a single
  heavy hitter. Running many copies of the protocol for unique heavy
  hitters in parallel, we can recover the list of heavy hitters. A
  careful analysis shows that the cost to privacy is essentially the
  same as running only a single copy of the underlying protocol. \rbnote{NEW: added citations from comp sensing literature}

  Along the way, we provide simpler and more private frequency-oracle
  protocols. Specifically, we show that the ``JL'' protocol of
  \citet{HKR10} can be made $(\eps,0)$-differentially private, and can
  be simplified to use computations in much smaller dimension (roughly,
  $O(n)$ instead of $\Omega(n^4\log d)$).

\paragraph{Lower Bounds on Error.} We show that, regardless of computation
  time and communication, every local $(\eps,\delta)$-DP protocol for
  frequency estimation has
  worst-case error $\Omega(\sqrt{\frac{\log (d)} {\eps^2 n}}))$
  as long as $\delta\ll 1/n$. This shows that our efficient protocols
  have optimal error.

  The instances that give rise to this lower bound are simple: one
  particular item $v$ (unknown to the algorithm) appears with
  frequency $\eta$, while the remaining inputs are chosen uniformly at
  random from $[d]\setminus\{v\}$. The structure of these instances
  has several implications. First, our lower bounds apply equally well
  to worst-case error (over data sets), and ``minimax error''
  (worst-case error over distributions in estimating the underlying
  distribution on data).

  Second, the accuracy of frequency estimation protocols must depend
  on the universe size $d$ in the local model,  even if one item appears much more
  frequently than all others. In contrast, in a centralized model,
  there are $(\eps,\delta)$-differentially private protocols that achieve
  error independent of the universe size, assuming only that there is a
  small gap (about $\frac{\log(1/\delta)}{\eps n}$) between
  the frequencies of the heaviest and second-heaviest hitters.

  The proof of our lower bounds adapts (and simplifies) a framework developed by
  \citet{DuchiJW13} for translating lower bounds on statistical
  estimation to the local privacy model. We make their framework more
  modular, and  show that it can be used  to prove lower bounds for
  $(\eps,\delta)$-differentially private protocols for $0<\delta<1/n$
  (in its original instantiation, it applied only for $\delta=0$). One
   lemma, possibly of independent interest, states that the mutual information between the input and
  output of a local protocol is at most $O(\eps^2+\frac{\delta}{\eps}
  \log(d\eps/\delta))$. In particular, the relaxation with $\delta>0$
  does not allow one to circumvent information-theoretic lower bounds
  unless $\delta$ is very large.

  \paragraph{1-bit Protocols Suffice for Local Privacy.} We show that
  a slight modification to the compression technique of McGregor et
  al. \cite[Theorem 14]{MMPRTV10} yields the following: in a public
  coin model (where the server and players have access to a common
  random string), every $(\eps,0)$-DP local protocol can be
  transformed so that each user sends only a single bit to the
  server. Moreover, the transformation is efficient under the
  assumption that one can efficiently compute conditional
  probabilities $Q(y|x)$ for the randomizers in the protocol.
  To our knowledge, all the local protocols in the literature (in
  particular, our efficient protocol for heavy hitters) satisfy this extra
  computability condition.

  The randomness of the public coins affects utility but not
  privacy in the transformed protocol; in particular, the coins may be generated by the untrusted
  server, by applying a pseudorandom function to the user's ID (if it
  is available), or by expanding a short seed sent by the user  using
  a pseudorandom generator.

  The transformation, following \cite{MMPRTV10}, is based on rejection
  sampling: the public coins are used to select a random sample from a
  fixed distribution, and a player uses his input to decide whether or
  not the sample should be kept (and used by the server) or
  ignored. This decision is transmitted as 1 bit to the server. Local
  privacy ensures that the rejection sampling procedure accepts with
  sufficiently large probability (and leaks little information about
  the input).

\subsection{Other Related Work}
\rbnote{Add quick reference to Ulfar's recent paper.}

In addition to the works mentioned so far on frequency estimation
\cite{MS06,HKR10,Rappor14,Rappor-unknowns15}, many papers have studied the complexity of
local private protocols for specific tasks. \asnote{Make a table if time}
\asnote{Add paragraph about $O(d)$-time algorithms, and algorithms for
  other utility measures. E.g. DJW $\ell_2$ and $\ell_1$ error.}

Most relevant here are the results of \citep{KLNRS08,DuchiJW13}
on learning and statistical estimation in the local model.
\citet{KLNRS08}  showed that when data are drawn i.i.d. from a
distribution, then every LDP learning algorithm can be simulated in
the statistical queries model \citep{Kearns98}. In particular, they
showed that learning parity and related functions requires an
exponential amount of data. Their simulation technique is the
inspiration for our communication reduction result.

Recently, \citet{DuchiJW13} studied a class of convex statistical
estimation problems, giving tight (minimax-optimal) error
guarantees. One of the local randomizers developed in \citep{DuchiJW13}
was the basis for the ``basic randomizer'' which is a building block
for our protocols. Moreover, our lower bounds are based
on the information-theoretic framework they establish.

Finally, our efficient protocols are based on ideas from the large
literature on streaming algorithms and compressive sensing (as were
the efficient protocols of \citet{HKR10}). For example, the use of
hashing to isolate unique ``heavy'' items appears  in
the context of sparse approximations to a vector's Fourier
representation \cite{GGIMS02} (and arguably that  idea has roots in
learning algorithms for Fourier coefficients such as \cite{Kushilevitz-Mansour91}).
This provides further
evidence of the close relationship between low-space algorithms and
differential privacy (see, e.g., \citep{panprivate,DNPR10,BlockiBDS12jl,KenthapadiKMM12,Upadhyay13}). 
\rbnote{NEW: added citations and brief description of the context}

\ifnum\full=1

\section{Building Blocks}\label{sec:upper}

\subsection{Useful Tools}\label{subsec:use-tools}

In this subsection, we will introduce some of the tools that we will use in our constructions.

First, we describe a basic randomizer (Algorithm~\ref{def:local-rand}) that will be used in our constructions as a tool to ensure that each user generates an $\eps$-differentially private report. This randomizer is a
more concise version of one of the randomizers in \citet{DuchiJW13}.


\mypar{Basic randomizer} Our basic randomizer $\R$ takes as input either an $m$-bit string represented by one of the vertices of the hypercube $\hypc$, or a special symbol represented by the all-zero $m$-length vector $\zeroB$. The randomizer $\R$ picks a bit $x_j$ at random from the input string $\bx$ (where $j$ is the index of the chosen bit), then it randomizes and scales $x_j$ to generate a bit $z_j\in\{-c_{\eps} \sqrt{m}, c_{\eps} \sqrt{m}\}$ (for some fixed $c_{\eps}=O(1/\eps)$). Finally, $\R$ outputs the pair $(j, z_j)$.  As will become clear later in our constructions, the $m$-bit input of $\R$ will be a unique encoding of one of the items in $\V$ whereas the special symbol $\zeroB$ will serve notational purposes to describe a special situation in our constructions when a user sends no information about its item.



\begin{algorithm}[h]
	\caption{$\R$: $\eps$-Basic Randomizer}
	\begin{algorithmic}[1]
		\REQUIRE $m$-bit string $\bx\in\hypc\cup\{\zeroB\}$, and the privacy parameter $\eps$.
		\STATE Sample $j\leftarrow [m]$ uniformly at random. 
\label{step:gen-rand-index}
        \IF{$\bx\neq \zeroB$}
                {\STATE Randomize $j$-th bit $x_j$ of the input $\bx\in\hypc$ as follows: \label{step:bit-random}
			    \begin{align}
			      &z_j\hspace{-.1cm}=\left\{\begin{array}{cc}
                                                        c_{\eps} m x_j & \text{ w.p. } \frac{e^{\eps}}{e^{\eps}+1} \\
					           -c_{\eps} m x_j & \text{ w.p. } \frac{1}{e^{\eps}+1}
					                 \end{array}\right.\nonumber
			    \end{align} 
			where $c_{\eps}=\frac{e^{\eps}+1}{e^{\eps}-1}=O\left(\frac{1}{\eps}\right)$. }
        \ELSE
	   {\STATE  Generate a uniform bit: $z_j\leftarrow\{-c_{\eps} \sqrt{m}, c_{\eps} \sqrt{m}\}$.}
 	\ENDIF
	\RETURN $\bz=\left(0, \ldots, 0, z_j, 0, \ldots, 0\right)\in \{-c_{\eps} \sqrt{m}, c_{\eps} \sqrt{m}\}^m$
        where $z_j$ is in the $j$-th position of $\bz$. (This
        output can be represented concisely by the pair $(j,
        z_j)$ using $\lceil\log m\rceil+1$ bits).\label{step:output_of_R}
	\end{algorithmic}
	\label{def:local-rand}
\end{algorithm}

\begin{thm}
$\R$ has the following properties:
\begin{enumerate}
\item $\R$ is $\eps$-$\ldp$ for every choice of the index $j$ \\ (that
  is, privacy depends only on the randomness in Step~\ref{step:bit-random}).
\item For every $\bx\in\hypc\cup\{\zeroB\}$, $\R(\bx)$ is an unbiased estimator of $\bx$. That is, $\E\left[\R(\bx)\right]=\bx$.  
\item $\R$ is computationally efficient (i.e., $\R$ runs in $O\left(m\right)$ time).
\end{enumerate}
\label{thm:prop-local-rand}
\end{thm}

As noted in Step~\ref{step:output_of_R} of the algorithm, we view the
output of $\R$ as a vector $\bz\in \re^m$ of the same length as the input
vector $\bx$. However, the output can be represented concisely by only
$\lceil\log m\rceil+1$ bits (required to describe the index $j$ and
$\z_j$). 

In some settings, we may compress this output to just 1 bit. This
comes from the fact that the privacy of $\R$ holds no matter how the index $j$
is chosen in Step~\ref{step:gen-rand-index}, so long as it is
independent of the input. (The randomness of $j$ \emph{is} important
for utility since it helps ensure that $\E\left[\R(\bx)\right]=\bx$.)
In particular, the randomness in the choice of $j$ may come from
outside the randomizer: it could be sent by the server, available as public
coins, or generated pseudorandomly from other information. In such
situations, the server receives $j$ through other channels and we may
represent the output using the single bit describing $z_j$.

\mypar{Johnson-Lindenstrauss Transform} Next, we go over the well-known Johnson-Lindenstrauss lemma that will be used to efficiently construct a private frequency oracle. This idea was originally used to provide an inefficient protocol for private estimation of heavy hitters in \cite{HKR10} (as opposed to providing just a private frequency oracle).


\begin{thm}[Johnson-Lindenstrauss lemma]
Let $0<c<1$ and $d\in\mathbb{N}$. Let $\mathcal{U}$ be any set of $t$ points in $\re^{d}$ and let $m\geq \frac{8\log(t)}{c^2}$. There exists a linear map $\Phi:\re^d\rightarrow\re^m$ such that $\Phi$ is approximately an isometric embedding of $\mathcal{U}$ into $\re^m$. Namely, for all $\bx,~\by\in \mathcal{U}$, we have
\begin{align}
(1-c)\ltwo{\bx-\by}^2\leq\ltwo{\Phi(\bx &- \by)}^2\leq (1+c)\ltwo{\bx-\by}^2\nonumber\\
\vert \langle \Phi\bx, \Phi\by\rangle - \langle \bx, \by\rangle\vert&\leq O\left(c\left(\ltwo{\bx}^2+\ltwo{\by}^2\right)\right) \nonumber
\end{align}
Moreover, any random $m\times d$ matrix with entries drawn i.i.d. uniformly from $\{-\frac{1}{\sqrt{m}}, \frac{1}{\sqrt{m}}\}$ enjoys this property with probability at least $1-\beta$ when $m=O\left(\frac{\log(t)\log(1/\beta)}{c^2}\right)$. Note that in such case, this matrix does not depend on the points in $\mathcal{U}$ (it only depends on the size of $\mathcal{U}$).
\label{thm:JL}
\end{thm}

\mypar{Basic tools from coding theory} Finally, we review some basics from coding theory that we will use in our efficient construction of private succinct histograms. For reasons that will become clear later, we will define a binary code of block length $m$ as a subset of $\hypc$ rather than $\{0, 1\}^m$.

\begin{defn}[A binary $(2^t, m, \zeta)$-code]
A binary $(2^t, m, \zeta)$-code is a pair of mappings $\left(\enc, \dec\right)$ where $\enc: \{1, ..., 2^t\}\rightarrow\hypc$ such that the set of the resulting vectors in $\hypc$, denoted by $\C$, satisfies the following constraint:
\begin{align}
&\min\limits_{\bx, \bx' \in\C}\ltwo{\bx-\bx'}\geq 2\sqrt{\zeta}\nonumber\\
\text{equivalently,  }\hspace{0.3cm}&\max\limits_{\bx, \bx' \in\C}\langle \bx, \bx'\rangle\leq 1-2\zeta\nonumber
\end{align}
and $\dec:\hypc\rightarrow \{1, ..., 2^t\}$ is some decoding rule that maps any given element in $\hypc$ to one of the codewords of $\C$.
\label{def:bin-code}
\end{defn}

The parameter $\zeta$ is known as the relative distance of the code. A binary $(2^t, m, \zeta)$-code can correct up to $\zeta/2$-fraction of errors. In other words, a binary $(2^t, m, \zeta)$-code has a decoder $\dec$ such that for any codeword $\bx\in\C$ and any erroneous version $\by\in\hypc$ of $\bx$ whose Hamming distance to $\bx$ is less than $m\zeta/2$, i.e.,
$$\sum_{j=1}^m\ind(\bx(j)\neq \by(j))<m\zeta/2$$
(or equivalently, $\langle \bx, \by\rangle > 1-\zeta$), we have $\dec(\by)=\bx$.

Moreover, for any $0<\zeta<1/2$, there is a construction of a binary $(2^t, O(t), \zeta)$-code with an efficient encoding and decoding algorithms. In fact, there are several constructions in coding theory literature that satisfy this property (for example, see \cite{venkat-thesis}).




\else

\section{Building Blocks}\label{sec:upper}

\subsection{The Basic Randomizer}

We describe a basic randomizer (Algorithm~\ref{def:local-rand}) that
will be used in our constructions as a tool to ensure that each user
generates an $\eps$-differentially private report. This randomizer is a
more concise version of one of the randomizers in \citet{DuchiJW13}.

Our basic randomizer $\R$ takes as input either an $m$-bit string represented by one of the vertices of the hypercube $\hypc$, or a special symbol represented by the all-zero $m$-length vector $\zeroB$. The randomizer $\R$ picks a bit $x_j$ at random from the input string $\bx$ (where $j$ is the index of the chosen bit), then it randomizes and scales $x_j$ to generate a bit $z_j\in\{-c_{\eps} \sqrt{m}, c_{\eps} \sqrt{m}\}$ (for some fixed $c_{\eps}=O(1/\eps)$). Finally, $\R$ outputs the pair $(j, z_j)$.  As will become clear later in our constructions, the $m$-bit input of $\R$ will be a unique encoding of one of the items in $\V$ whereas the special symbol $\zeroB$ will serve notational purposes to describe a special situation in our constructions when a user sends no information about its item.

\begin{algorithm}[h]
	\caption{$\R$: $\eps$-Basic Randomizer}
	\begin{algorithmic}[1]
		\REQUIRE $m$-bit string $\bx\in\hypc\cup\{\zeroB\}$, and the privacy parameter $\eps$.
		\STATE Sample $j\leftarrow [m]$ uniformly at random. 
\label{step:gen-rand-index}
        \IF{$\bx\neq \zeroB$}
                {\STATE Randomize $j$-th bit $x_j$ of the input $\bx\in\hypc$ as follows: \label{step:bit-random}
			    \begin{align}
			      &z_j\hspace{-.1cm}=\left\{\begin{array}{cc}
                                                        c_{\eps} m x_j & \text{ w.p. } \frac{e^{\eps}}{e^{\eps}+1} \\
					           -c_{\eps} m x_j & \text{ w.p. } \frac{1}{e^{\eps}+1}
					                 \end{array}\right.\nonumber
			    \end{align} 
			where $c_{\eps}=\frac{e^{\eps}+1}{e^{\eps}-1}=O\left(\frac{1}{\eps}\right)$. }
        \ELSE
	   {\STATE  Generate a uniform bit: $z_j\leftarrow\{-c_{\eps} \sqrt{m}, c_{\eps} \sqrt{m}\}$.}
 	\ENDIF
	\RETURN $\bz=\left(0, \ldots, 0, z_j, 0, \ldots, 0\right)\in \{-c_{\eps} \sqrt{m}, c_{\eps} \sqrt{m}\}^m$
        where $z_j$ is in the $j$-th position of $\bz$. (This
        output can be represented concisely by the pair $(j,
        z_j)$ using $\lceil\log m\rceil+1$ bits).\label{step:output_of_R}
	\end{algorithmic}
	\label{def:local-rand}
\end{algorithm}

\break
\begin{thm}
$\R$ has the following properties:
\begin{enumerate}
\item $\R$ is $\eps$-$\ldp$, for every choice of the index $j$ (that
  is, privacy depends only on the randomness in Step~\ref{step:bit-random}).
\item For every $\bx\in\hypc\cup\{\zeroB\}$, $\R(\bx)$ is an unbiased estimator of $\bx$. That is, $\E\left[\R(\bx)\right]=\bx$.  
\item $\R$ is computationally efficient (i.e., $\R$ runs in $O\left(m\right)$ time).
\end{enumerate}
\label{thm:prop-local-rand}
\end{thm}

As noted in Step~\ref{step:output_of_R} of the algorithm, we view the
output of $\R$ as a vector $\bz\in \re^m$ of the same length as the input
vector $\bx$. However, the output can be represented concisely by only
$\lceil\log m\rceil+1$ bits (required to describe the index $j$ and
$\z_j$). 

In some settings, we may compress this output to just 1 bit. This
comes from the fact that the privacy of $\R$ holds no matter how the index $j$
is chosen in Step~\ref{step:gen-rand-index}, so long as it is
independent of the input. (The randomness of $j$ \emph{is} important
for utility since it helps ensure that $\E\left[\R(\bx)\right]=\bx$.)
In particular, the randomness in the choice of $j$ may come from
outside the randomizer: it could be sent by the server, available as public
coins, or generated pseudorandomly from other information. In such
situations, the server receives $j$ through other channels and we may
represent the output using the single bit describing $z_j$.

\fi

\ifnum\full=1

\subsection{A Private Frequency Oracle Construction}\label{subsec:JL}

We give here an efficient construction of a private frequency oracle
based on Johnson-Lindenstrauss projections. Our construction follows
almost the same lines of the construction of \cite{HKR10}. Our version
differs in three respects. First, we use the construction only to provide a frequency oracle as opposed to identifying and estimating the frequency of heavy hitters. For that purpose, the construction is computationally efficient. The second difference is in the local randomization step at each user. Here, each user $i\in[n]$ uses an independent copy of the basic randomizer $\R_i$ given by Algorithm~\ref{def:local-rand} (as opposed to adding noise as in \cite{HKR10}). This gives us pure $\eps$-differential privacy guarantee (as opposed to $(\eps, \delta)$ in \cite{HKR10}). The third difference is that computations are carried out in much smaller dimension, namely $O(n)$ as opposed to $\Omega(n^4\log(d))$ in \cite{HKR10}.

Given our private frequency oracle, we give a simple efficient algorithm that, for any given input $v\in\V$, uses the frequency oracle to obtain a private estimate $\hf(v)$ of the frequency $f(v)$ of the item $v$.

Let $\Phi$ denote an $m\times d$ random projection matrix as in Theorem~\ref{thm:JL} with $m=\frac{\log(d+1)\log(2/\beta)}{\gamma^2}$ and $\gamma=\sqrt{\frac{\log(2d/\beta)}{\eps^2 n}}$ where $\beta>0$ is an input parameter to our algorithm that, affects the confidence level of our error guarantee (but not the privacy guarantee). In our protocol below, we assume the existence of a source of randomness $\genproj$ that on input integers $m, d>0$ generates an instance of $\Phi$. The output $\Phi$ of $\genproj$ is assumed to be public, that is, shared by all parties in the protocol (the users and the server).  We note that there are efficient constructions for $\genproj$ that generates a succinct description of $\Phi$ that is much less than $md$ when the columns of the projection matrix $\Phi$ are $k$-wise independent for $k<<d$. For our construction, it suffices for the columns of $\Phi$ to be $n$-wise independent (namely, it will still satisfy the conditions in Theorem~\ref{thm:JL}). Hence, the amount of randomness generated by $\genproj$ (describing $\Phi$) is $O(mn)$ in such case.

We denote the $i$-th standard basis vector in $\re^{d}$ by $\be_i$. The construction protocol of a private frequency is described below in Algorithm~\ref{Alg:FO}.

\begin{algorithm}[htb]
	\caption{$\prot$-$\fo$: $\eps$-$\ldp$ Frequency Oracle Protocol}
	\begin{algorithmic}[1]
		\REQUIRE Users' inputs $\{v_i\in\V: i\in[n]\}$, the privacy parameter $\eps$, and the confidence parameter $\beta>0$.
        \STATE  $\gamma\leftarrow \sqrt{\frac{\log(2d/\beta)}{\eps^2 n}}$.
        \STATE $m\leftarrow \frac{\log(d+1)\log(2/\beta)}{\gamma^2}$.
	\STATE $\Phi\leftarrow\genproj(m, d)$.
        \FOR{ Users $i=1$ to $n$}
            {\STATE User $i$ computes $\bz_i=\R_i\left(\Phi\be_{v_i}, \eps\right)$.}\label{Q-fo}
	   {\STATE User $i$ sends $\bz_i$ to the server.}
        \ENDFOR
	\STATE Server computes $\bar{\bz}=\frac{1}{n}\sum_{i=1}^n\bz_i.$\label{step:agg}
	\STATE $\fo\leftarrow \left(\Phi, \bbz\right).$
	\RETURN  $\fo$
	\end{algorithmic}
	\label{Alg:FO}
\end{algorithm}

Note that $m=O(n)$. Hence, the length of each user's report is $O(\log(m))=O(\log(n))$. Moreover, as noted above we only need $O(mn)=O(n^2)$ random bits to generate $\Phi$, thus, $\genproj$ runs in time $O(n^2)$. Also, each basic randomizer is efficient, i.e., runs in $O(m)=O(n)$ time (Part~3 of Theorem~\ref{thm:prop-local-rand}). Hence, one can easily verify that the construction is computationally efficient.



In Algorithm~\ref{Alg:use-fo} below, we show that, for any given fixed item $v\in\V$, $\fo$ can be used to efficiently give an estimate $\hf{v}$ of $f(v)$.

\begin{algorithm}[htb]
	\caption{$\A_{{\sf FO}}$: $\eps$-$\ldp$ Frequency Estimator Based on $\fo$}
	\begin{algorithmic}[1]
		\REQUIRE Data structure $\fo=\left(\Phi, \bbz\right)$ (the frequency oracle), an item $v\in\V$ whose frequency to be estimated.
       \RETURN $\hf(v)=\langle \Phi\be_{v}, \bar{\bz} \rangle$.
\end{algorithmic}
	\label{Alg:use-fo}
\end{algorithm}


The privacy and utility guarantees of the frequency oracle constructed by $\prot$-$\fo$ above are given in the following theorems.

\begin{thm}[Privacy of $\fo$]
The construction of the frequency oracle $\fo$ given by Algorithm~\ref{Alg:FO} is $\eps$-differentially private.
\label{thm:privacy-FO}
\end{thm}

\begin{proof}
The proof follows directly from part~1 of Theorem~\ref{thm:prop-local-rand}.
\end{proof}

\begin{thm}[Error of $\fo$]
Let $\eps>0$. For any set of users items $\{v_1, ..., v_n\}$ and any $\beta>0$, the error due to $\fo$ constructed by Algorithm~\ref{Alg:FO} is bounded as
$$\err\left(\f; \fo\right)\triangleq\max\limits_{v\in\V}\vert\hf(v)-f(v)\vert=O\left(\frac{1}{\eps}\sqrt{\frac{\log(d/\beta)}{n}}\right)$$
with probability at least $1-\beta$ over the randomness of the projection $\Phi$ and the basic randomizers $\R_i, i\in[n]$, where $\hf(v)$ denote the output of Procedure $\A_{{\sf FO}}$ (given by Algorithm~\ref{Alg:use-fo} above) on an input $v$.
\label{thm:error-FO}
\end{thm}

\begin{proof}
The proof relies on the good concentration behavior of the inner product between the aggregate measurement $\bbz$ and any vector $\by\in\hypc$. This is formalized in the following claim. 

\begin{claim}\label{conc-agg-inner-prod}
Let $\beta>0$. Let $\bx_1, \ldots, \bx_n, by\in\hypc$ and let $\bz_i=\R_i(\bx_i, \eps)$ where $\R_i, i\in[n]$ are independent copies of our basic randomizer (Algorithm~\ref{def:local-rand}). Then, with probability at least $1-\beta$, we have 
$$\left\vert\frac{1}{n}\sum_{i=1}^n\langle \bz_i-\bx_i, \by\rangle\right\vert=O\left(\frac{1}{\eps}\sqrt{\frac{\log(1/\beta)}{n}}\right).$$
\end{claim}

To prove this claim, we first observe that $\langle \bz_i, \by\rangle, ~i\in[n],$ is a sequence of independent random variables taking values in $\left\{-O\left(\frac{1}{\eps}\right), O\left(\frac{1}{\eps}\right)\right\}$. Also, from the second property of our basic randomizer (Theorem~\ref{thm:prop-local-rand}), we have $\E\left[\langle \bz_i, \by \rangle\right]=\langle \bx_i, \by \rangle$. Putting these together, then by Hoeffding's inequality our claim follows.

To prove the theorem, observe that the error of construction $\fo$ can be written as
\begin{align}
\max\limits_{v\in\V}\vert\hf(v)-f(v)\vert &=\max\limits_{v\in\V}\left\vert \langle \bar{\bz}, \Phi\be_{v} \rangle - \langle \frac{1}{n}\sum_{i=1}^n\be_{v_i}, \be_{v}\rangle\right\vert\nonumber\\
&=\max\limits_{v\in\V}\left\vert \langle \bar{\bz} -\E[\bbz], \Phi\be_{v} \rangle +\langle \Phi\left(\frac{1}{n}\sum_{i=1}^n\be_{v_i}\right), \Phi\be_{v} \rangle - \langle \frac{1}{n}\sum_{i=1}^n\be_{v_i}, \be_{v}\rangle\right\vert\label{ineq:FO-err-unbiased}\\
&\leq \max\limits_{v\in\V}\left\vert \frac{1}{n}\sum_{i=1}^n\langle \bz_i -\E[\bz_i], \Phi\be_{v} \rangle\right\vert + \max\limits_{v\in\V}\left\vert \langle \Phi\left(\frac{1}{n}\sum_{i=1}^n\be_{v_i}\right), \Phi\be_{v} \rangle - \langle \frac{1}{n}\sum_{i=1}^n\be_{v_i}, \be_{v}\rangle \right\vert\label{ineq:FO-err-triangle}
\end{align}
where (\ref{ineq:FO-err-unbiased}) follows from Part~2 of Theorem~\ref{thm:prop-local-rand}, and (\ref{ineq:FO-err-triangle}) follows from the linearity of the inner product and the triangle inequality.

Now, by Johnson-Lindenstrauss lemma (Theorem~\ref{thm:JL}), with probability at least $1-\beta/2$, the second term is bounded by $\gamma\cdot O(1)=O(\frac{1}{\eps}\sqrt{\frac{\log(d/\beta)}{n}})$. We next consider the first term. Fix $v\in\V$. By Claim~\ref{conc-agg-inner-prod} above, with probability at least $1-\frac{\beta}{2d}$, we have
$$\left\vert \frac{1}{n}\sum_{i=1}^n\langle \bz_i -\E[\bz_i], \Phi\be_{v} \rangle\right\vert\leq O\left(\frac{1}{\eps}\sqrt{\frac{\log(d/\beta)}{n}}\right).$$
Hence, by the union bound,  with probability at least $1-\beta/2$, the first term of (\ref{ineq:FO-err-triangle}) is bounded by $O\left(\frac{1}{\eps}\sqrt{\frac{\log(d/\beta)}{n}}\right)$. Thus, with probability at least $1-\beta$, the error of $\fo$ is bounded as in Theorem~\ref{thm:error-FO}.
\end{proof}
\mypar{Note} The above upper bound is shown to be tight by our result in Section~\ref{sec:lower}.

\else

\subsection{A Private Frequency Oracle}\label{subsec:JL}

We give an efficient private frequency oracle that follows almost the same lines of the construction of \cite{HKR10}. Our protocol differs from \cite{HKR10} in three respects. First, we use the construction only to provide a frequency oracle as opposed to identifying and estimating the frequency of heavy hitters. For that purpose, the construction is computationally efficient. The second difference is in the local randomization step at each user. Here, each user $i\in[n]$ uses an independent copy of the basic randomizer $\R_i$ given by Algorithm~\ref{def:local-rand} (as opposed to adding noise as in \cite{HKR10}). This gives us pure $\eps$-differential privacy guarantee (as opposed to $(\eps, \delta)$ in \cite{HKR10}). The third difference is that computations are carried out in much smaller dimension, namely $O(n)$ as opposed to $\Omega(n^4\log(d))$ in \cite{HKR10}.

The description of our frequency oracle construction protocol is given in the full version \cite{fulvBS15}. We refer to such protocol as $\prot$-$\fo$. This protocol outputs a frequency oracle $\fo$ which is composed of two objects; a \emph{succinct} description of a binary matrix $\Phi$ whose columns represent encodings of each item in $\V$, and, an aggregate measurement of users reports $\bbz$. Given our private frequency oracle, there is a simple efficient algorithm $\A_{{\sf FO}}$ (see the full version \cite{fulvBS15} for a complete description) that, for any given input $v\in\V$, uses the frequency oracle to obtain a private estimate $\hf(v)$ of the frequency $f(v)$ of the item $v$ by simply computing the inner product between the encoding of $v$ under $\Phi$ and an aggregate measurement $\bbz$, which is the average of users' reports.

In protocol $\prot$-$\fo$, the length of encoding of an item under $\Phi$ is $O(n)$, the report length of each user is $O(\log(n))$, and the total amount of randomness required to generate $\Phi$ is $O(n^2)$ random bits. Also, each basic randomizer is efficient, i.e., runs in $O(n)$ steps (Part~3 of Theorem~\ref{thm:prop-local-rand}). Hence, the construction is computationally efficient. Also, generating an estimate $\hf(v)$ for a given item $v$ using $\A_{{\sf FO}}$ takes only $O(n)$ time.



%


The privacy and utility guarantees of the frequency oracle constructed by $\prot$-$\fo$ are given in the following theorems.

\begin{thm}
The construction of the frequency oracle $\fo$ given by Protocol $\prot$-$\fo$ is $\eps$-differentially private.
\label{thm:privacy-FO}
\end{thm}


\begin{thm}
Let $\eps>0$. For any set of users items $\{v_1, ..., v_n\}$ and any $\beta>0$, the error due to $\fo$ constructed by Protocol $\prot$-$\fo$ is bounded as
$$\err\left(\f; \fo\right)\triangleq\max\limits_{v\in\V}\vert\hf(v)-f(v)\vert=O\left(\frac{1}{\eps}\sqrt{\frac{\log(d/\beta)}{n}}\right)$$
with probability at least $1-\beta$ over the randomness of the projection $\Phi$ and the basic randomizers $\R_i, i\in[n]$, where $\hf(v)$ denote the output of Procedure $\A_{{\sf FO}}$ on an input $v$.
\label{thm:error-FO}
\end{thm}
The above upper bound is asymptotically tight (Section~\ref{sec:lower}). The theorem's proof relies on the concentration of the inner product between the aggregate $\bbz$ and the encoding of any given item under the encoding matrix $\Phi$. See the full version \cite{fulvBS15} for details.  

\fi


\ifnum\full=1
\section{Efficient Error-Optimal Construction of Private Succinct Histograms}\label{sec:eff-protocol}
\else
\section{Efficient Construction \\ with Optimal Error}
\fi
In this section, our goal is to construct an efficient private succinct histogram using the private frequency oracle given in the previous subsection together with other tools. In Section~\ref{subsec:promise-prob}, we first give a construction for a simpler problem that we call the unique heavy hitter problem. Then, in Section~\ref{subsec:eff-construction}, we give a reduction from this problem to the general problem.

\subsection{The Unique Heavy Hitter Problem}\label{subsec:promise-prob}

In the unique heavy hitter problem, we are given the promise that at least an $\eta$ fraction of the $n$
users hold the same item $v^*$ for some $v^*\in\V$  unknown to
the server (here $\eta$ is a parameter of the promise), and that all other users  hold a special symbol $\bot$, representing ``no item''.

Our goal is to obtain an efficient construction of a private succinct
histogram under this promise, for as small a value $\eta$ as possible. We will take
$\eta$ to be at least
$\frac{C}{\eps}\sqrt{\frac{\log(d)}{n}}$ for a universal constant
$C>0$. Our protocol is differentially private on all inputs. Under
the promise, with high probability, it outputs the correct $v^*$ together
with an estimate $\hf(v^*)$ of the frequency $f(v^*)$. 

The main idea of the protocol is to first encode user's items with an
error-correcting code and randomize the resulting codeword before
sending it to the server. The redundancy in the code allows the server
learn the unknown item $v^*$ from the noisy reports.

We require an efficiently encodable and
decodable binary $(d, m, \zeta)$-code (of $d$ codewords, block length $m$, and relative distance $\zeta$) where $m=O(\log(d))$ with constant rate (so that
$m=O(\log(d))$) and constant relative minimum distance
$\zeta\in(0,1/2)$, say $\zeta=1/4$. (We do not require the rate
$\frac{\log{d}}{m}$ or relative distance $\zeta$ to be optimal;
these quantities will affect the constants in the error of our
construction but not the asymptotic behavior.) There are several known
constructions of such codes in the literature (see
\cite{venkat-thesis} for
examples). 
Fix one such code, denoted $\code$, with associated encoder $\bc$ and
decoder $\dec$. The code is part of the protocol and so is known to
all parties. For convenience, we represent codewords as
points in the unit-radius hypercube $\hypc$. 

Each user $i$  first encodes its item $v_i$ to obtain
$\bx_i=\bc(v_i)\in\hypc$, then runs the basic randomizer $\R_i$ (given by
Algorithm~\ref{def:local-rand}) on $\bx_i$ to obtain the 
report $\bz_i$. Users that have no item, i.e., users with input
$\bot$,  feed the zero vector $\bx_i=\zeroB$ to the basic randomizer.

The server aggregates the reports by computing
$\bbz=\frac{1}{n}\sum_{i=1}^n\bz_i$, and then decodes $\bbz$
to obtain the encoding $\bx$ of $v^*$. One may not be able to feed $\bbz$
directly to the decoding algorithm $\dec$ of $\code$ since  $\bbz$
will not, in general, be a vertex of the hypercube $\hypc$. Instead, the
server first 
rounds the aggregated signal $\bbz$ to the nearest point $\by$ in the
hypercube before running $\dec$. We argue that the combination of
noise from randomization and the rounding step produces a vector $\by$
that is sufficiently close to $\bx$ with high probability. 



Algorithm~\ref{Alg:PP} precisely describes our construction for the promise problem.  The protocol is computationally efficient, i.e., the total computational cost is $\poly(\log(d), n)$ since $\code=(\bc, \dec)$ runs in time $\poly(\log(d))$ and each basic randomizer $\R_i$ runs in time $O(\log(d))$. In fact, the computational cost at each user does not depend on $n$. Also, we note that the users' reports are succinct, namely, the report length is $O\left(\log\left(\log(d)\right)\right)$ bits.

\begin{algorithm}[htb]
	\caption{$\prot^{\sf PP}$-$\sh_{\sf PP}$: $\eps$-$\ldp$ Succinct Histogram Protocol under the Promise}
	\begin{algorithmic}[1]
		\REQUIRE Users' inputs $\{v_i\in\V\cup\{\bot\}: i\in[n]\}$, the privacy parameter $\eps$, and the confidence parameter $\beta>0$.
        \FOR{ Users $i=1$ to $n$}
            {\STATE If $v_i\neq \bot$, then user $i$ encodes its item:
              $\bx_i=\bc(v_i)
              $. Else, user $i$ sets $\bx_i=\zeroB$.}\label{Q-pp-1}
	   {\STATE User $i$ computes its private report: $\bz_i=\R_i\left(\bx_i, \eps\right)$.}\label{Q-pp-2}
	  {\STATE User $i$ sends $\bz_i$ to the server.}
        \ENDFOR
	\STATE Server computes $\bar{\bz}=\frac{1}{n}\sum_{i=1}^n\bz_i$.
{\STATE \label{step:round}
Server computes $\by$ by rounding $\bbz$ to $\hypc$. That is,
  for each $j=1,...,m$, \\
	 $y_j=\begin{cases}
          \frac{1}{\sqrt{m}} & \text{ if } \bar{z}_j\geq 0 \,, \text{ and}\\
          -\frac{1}{\sqrt{m}} & \text{ otherwise, } 
        \end{cases}$\qquad 
	where $\bar{z}_j$ denotes the $j$-th entry of $\bbz$.}
	\STATE Server decodes $\by$ into an estimate for the common
        item $\hat{v}=\dec(\by)$ \\ and computes a frequency estimate $\hf(\hat{v})=\langle \bc(\hat{v}), \bbz \rangle$.
	\RETURN $\left(\hat{v},~\hf(\hat{v})\right)$.
	\end{algorithmic}
	\label{Alg:PP}
\end{algorithm}


\begin{thm}[Privacy of $\sh_{\sf PP}$]
The construction of the succinct histogram~$\sh_{\sf PP}$ given by Algorithm~\ref{Alg:PP} is $\eps$-differentially private.
\label{thm:privacy-sh-promise}
\end{thm}
\begin{proof}
Privacy follows directly from the $\eps$-differential privacy of the basic randomizers $\R_i, i\in[n]$ (Part~1 of Theorem~\ref{thm:prop-local-rand}).
\end{proof}

To analyze utility, we first isolate the guarantee provided by the rounding step. Let $\S_m=\{\bw\in\re^m:\ltwo{\bw}=1\}$ denote the $m$-dimensional unit sphere.

\begin{lem}
Let $\bz\in\S_{m}$ be such that there is a codeword of $\code$,
$\bx\in\C$, with $\langle \bz, \bx\rangle > 1-\zeta/4$. Let $\by$ be
vector in the hypercube $\hypc$ obtained by rounding each entry $z_j$ of
$\bz$ to $sign(z_j)/\sqrt{m}$. Then the Hamming distance between  $\by$ and $\bx$ is less than $m\zeta/2$, i.e., $\sum_{j=1}^m \ind(y_j\neq x_j)<m\zeta/2$.
\label{lem:btr}
\end{lem}
\begin{proof}
Since $\bz$ and $\bx$ are unit vectors, the distance $\ltwo{\bz-\bx}$
satisfies 
$$\ltwo{\bz-\bx}^2 = \ltwo{\bz}^2+\ltwo{\bx}^2 -2\langle\bz,\bx\rangle \leq \zeta/2.$$ 
The vectors $\bx$ and $\by$ disagree in coordinate $j$ only if $\vert z_j -
x_j\vert\geq\frac{1}{\sqrt{m}}$. There can be at most
$m\zeta/2$ such coordinates, since each contributes at least $\frac 1
m$ to $\ltwo{\bz-\bx}^2$.
Thus, the Hamming distance between $\by$ and $\bx$ is $\sum_{j=1}^m \ind(y_j\neq x_j)<m\zeta/2$ completing the proof.
\end{proof}

\begin{thm}[Error of $\sh_{\sf PP}$ under the promise]
Let $\eps>0$. Suppose that the conditions in the above promise are true for some common item $v^*\in\V$. For any $\beta>0$, there is a setting of $\eta=O\left(\frac{1}{\eps}\sqrt{\frac{\log(d)\log\left(1/\beta\right)}{n}}\right)$ in the promise such that, with probability at least $1-\beta$, \mbox{Protocol}~$\prot^{\sf PP}$-$\sh_{\sf PP}$ publishes the right item $v^*$ and the frequency estimation error is bounded by
$$\max\limits_{v\in\V}\vert\hf(v)-f(v)\vert=O\left(\frac{1}{\eps}\sqrt{\frac{\log(1/\beta)}{n}}\right)\,.$$
\label{thm:error-sh-promise}
\end{thm}
\rbnote{slightly modified the proof to reflect the change of rounding
  alg.$\btr$}\asnote{inlined the rounding step to simplify the
  algorithm. Rewrote parts of the proof.}
\begin{proof}
Consider the conditions of the promise. Let $v^*\in\V$ be the unique heavy hitter  (occurring with frequency at least $\eta$). Let $\beta>0$. Given Lemma~\ref{lem:btr}, to show that the protocol above recovers the correct item $v^*$ with probability at least $1-\beta/2$, it suffices to show that, with probability at least $1-\beta/2$, we have
\begin{align}
\langle \bc(v^*), \frac{\bbz}{\ltwo{\bbz}} \rangle&>1-\zeta/4.\nonumber
\end{align}
Note that the rounding step (Step \ref{step:round} in Algorithm~\ref{Alg:PP}) would produce the
same output whether it was run with $\bbz$ or its normalized counterpart $\bbz/\ltwo{\bbz}$.

By the promise, we have
$$\bbz=\frac{1}{n}\sum_{i=1}^n\bz_i=\frac{1}{n}\sum_{i\in\T}\R_i\left(\bc(v^*)\right)+\frac{1}{n}\sum_{i\in[n]\setminus\T}\R_i\left(0\right)$$
where $\T$ denotes the set of users having the item $v^*$. (Note that $\frac{\vert \T\vert}{n}=f(v^*) \geq \eta$).

First, we consider $\ltwo{\bbz}$. Since for every $i\in[n]$, $\R_i$ is unbiased (Part~2 of Theorem~\ref{thm:prop-local-rand}),  we have $\ltwo{\E[\bbz]}=f(v^*)$. Using the triangle inequality, we get $\ltwo{\bbz}\leq f(v^*)+\ltwo{\bbz-\E\left[\bbz\right]}$.
Next, we obtain an upper bound on $\ltwo{\bbz-\E\left[\bbz\right]}$. Note that $\bz_i, i=1, .., n,$ are independent and that for every $i\in[n]$, $\ltwo{\bz_i}=O\left(\frac{\sqrt{m}}{\eps}\right)$ with probability $1$. Applying McDiarmid's inequality \cite{mcdiarmid}, with probability at least $1-\beta/4$, we have $\ltwo{\bbz-\E\left[\bbz\right]}\leq O\left(\frac{1}{\eps}\sqrt{\frac{m\log(1/\beta)}{n}}\right)$. Thus, with probability at least $1-\beta/4$, $\ltwo{\bbz}$ is bounded by
\begin{align}
\ltwo{\bbz}&\leq f(v^*)+O\left(\frac{1}{\eps}\sqrt{\frac{m\log(1/\beta)}{n}}\right)\label{bound:pp-bbz}
\end{align}
Next, we consider $\langle  \bc(v^*), \bbz\rangle$. Observe that $\langle  \bc(v^*), \bbz\rangle=$
\ifnum\full=1
\begin{align}
&f(v^*)+\frac{1}{n}\sum_{i\in\T}\langle \bc(v^*), \R_i(\bc(v^*))- \bc(v^*)\rangle +\frac{1}{n}\sum_{i\in[n]\setminus\T}\langle \bc(v^*), \R_i(\zeroB)\rangle\nonumber
\end{align}
\else
\begin{align}
&f(v^*)+\frac{1}{n}\sum_{i\in\T}\langle \bc(v^*), \R_i(\bc(v^*))- \bc(v^*)\rangle \nonumber\\
+&\frac{1}{n}\sum_{i\in[n]\setminus\T}\langle \bc(v^*), \R_i(\zeroB)\rangle\nonumber
\end{align}
\fi
By the tail properties of the distribution of the second and third terms (following Claim~\ref{conc-agg-inner-prod}), we can show that with probability at least $1-\beta/4$, we have
\begin{align}
\langle  \bc(v^*), \bbz\rangle&\geq f(v^*) - O\left(\frac{1}{\eps}\sqrt{\frac{\log(1/\beta)}{n}}\right)\label{bound:pp-innerprod}
\end{align}
\rbnote{express the ratio in terms of $f(v^*)$ then use the fact $\eta\leq f(v^*)$ to bound the ratio from below.}
Putting (\ref{bound:pp-bbz}) and (\ref{bound:pp-innerprod}) together, then, with probability at least $1-\beta/2$, we have 
$$\langle \bc(v^*), \frac{\bbz}{\ltwo{\bbz}} \rangle\geq \frac{\eta-O\left(\frac{1}{\eps}\sqrt{\frac{\log(1/\beta)}{n}}\right)}{\eta+O\left(\frac{1}{\eps}\sqrt{\frac{m\log(1/\beta)}{n}}\right)}$$
where we use the fact that $\eta\leq f(v^*)$ and assume that the numerator in the right-hand side is positive.

Since $m=O(\log(d))$, then there is a constant $\alpha_{\zeta}$ that depends on $\zeta$ such that if we set $\eta=\alpha_{\zeta}\frac{1}{\eps}\sqrt{\frac{\log(d)\log\left(1/\beta\right)}{n}}$, then the above ratio is greater than $1-\zeta/4$. This proves that there is a setting of $\eta=O\left(\frac{1}{\eps}\sqrt{\frac{\log(d)\log(1/\beta)}{n}}\right)$ such that construction $\prot^{\sf PP}$-$\sh_{\sf PP}$ outputs $\hat{v}=v$ with probability at least $1-\beta/2$.


Now, conditioned on correct decoding, for all $v\neq v^*$, the estimate $\hf(v)$ is implicitly assumed to be zero (which is perfectly accurate in this case). Thus, it remains to inspect $\hf(v^*)$. Observe that
\ifnum\full=1
\begin{align}
&\vert \hf(v^*)-f(v^*)\vert=\left\vert \langle \bc(v^*), \frac{\bbz}{\ltwo{\bbz}} \rangle- f(v^*) \right\vert\leq\left\vert\frac{1}{n}\sum_{i\in\T}\langle \bc(v^*), \R_i(\bc(v^*))- \bc(v^*)\rangle\right\vert +\left\vert\frac{1}{n}\sum_{i\in[n]\setminus\T}\langle \bc(v^*), \R_i(\zeroB)\rangle\right\vert\nonumber
\end{align}
\else
\begin{align}
&\vert \hf(v^*)-f(v^*)\vert=\left\vert \langle \bc(v^*), \frac{\bbz}{\ltwo{\bbz}} \rangle- f(v^*) \right\vert\nonumber\\
\leq&\left\vert\frac{1}{n}\sum_{i\in\T}\langle \bc(v^*), \R_i(\bc(v^*))- \bc(v^*)\rangle\right\vert +\left\vert\frac{1}{n}\sum_{i\in[n]\setminus\T}\langle \bc(v^*), \R_i(\zeroB)\rangle\right\vert\nonumber
\end{align}
\fi
Again, by the tail properties of the sums above, with probability at least $1-\frac{\beta}{2}$, we conclude that $\vert \hf(v^*)-f(v^*)\vert\leq O\left(\frac{1}{\eps}\sqrt{\frac{\log(2/\beta)}{n}}\right)$.

Therefore, with probability at least $1-\beta$, protocol $\prot^{\sf PP}$-$\sh_{\sf PP}$ recovers the correct common item $v^*$ and the estimation error that is bounded by $O\left(\frac{1}{\eps}\sqrt{\frac{\log(1/\beta)}{n}}\right)$.
\end{proof}

\ifnum\full=0

\subsection{Efficient Construction for the General \\Problem}\label{subsec:eff-construction}
In this section, we provide an efficient construction of private succinct histograms with optimal error for the general setting of the problem using the two protocols discussed in the previous sections as sub-protocols.


In the promise (unique heavy hitter) problem, the main advantage was the lack of interference from the users who do not hold the heavy hitter $v^*$ in question. The main idea here is to obtain a reduction in which we create the conditions of the promise problem separately for each heavy hitter $v^*\in\V$ such that the extra computational cost is at most a small $\poly(n)$ factor. To do this, we \emph{hash} each item $v\in\V$ into one of $K$ separate \emph{parallel} channels such that users holding the same item will transmit their reports in the same channel. Each user, in the remaining $K-1$ channels, will simulate an ``idle'' user with item $\bot$ as in the promise problem. By choosing $K$ sufficiently large, and repeating the protocol in \emph{parallel} for $T$ times\footnote{That is, the total number of parallel channels is $KT$. In each group of $K$ channels, a fresh hash seed is used.}, we can guarantee that, with high probability, \emph{every} heavy hitter $v^*\in\V$ gets assigned to an interference-free channel. Hence, by using an error-optimal construction for the promise problem like $\sh_{\sf PP}$ in each one of these channels, we eventually obtain a list of at most $KT$ items such that, with high probability, all the heavy hitters will be on that list. However, this list may also contain other erroneously decoded items due to hash collisions and we do not know which items on the list are the heavy hitters. To overcome this, in a separate parallel channel of the protocol, we run a frequency oracle protocol (like $\prot$-$\fo$) and use the resulting frequency oracle to estimate the frequencies of all the items on that list, then output all the items whose estimated frequencies are above $\eta$ together with their estimated frequencies.

For the purpose of this construction, it suffices to use a pairwise independent hash function $\hash:$ $\{0, 1\}^{\ell}\times\V\rightarrow [K]$ whose first input is a random seed $s\in\{0,1\}^{\ell}$. Choices of $\ell$ and $K$ are given in the protocol description (Algorithm~\ref{Alg:SH-eff}). All users and the server are assumed to have access to $\hash$. We use a source of \emph{public} randomness $\gen$ that, on an input integer $\ell>0$, generates a \emph{seed} for $\hash$, which is a random uniform string from $\{0,1\}^{\ell}$ that is seen by everyone\footnote{We may also think of $\gen$ as being run at the server which then announces the resulting random string to all the users.}. 

\else
\subsection{Efficient Construction for the General Problem}\label{subsec:eff-construction}
In this section, we provide an efficient construction of private succinct histograms for the general setting of the problem using the two protocols discussed in the previous sections as sub-protocols. Namely, our construction uses an efficient private frequency oracle like $\fo$ given in Section~\ref{subsec:JL} and an efficient private succinct histogram for the promise problem like $\sh_{\sf PP}$ given in Section~\ref{subsec:promise-prob}. Our construction is modular and does not depend on the internal structure of the construction protocols or the data structure. Our construction is efficient and succinct as long as the construction of such objects is efficient and succinct. Moreover, our succinct histogram is shown to be error-optimal if the aforementioned objects satisfy the guarantees of Theorems~\ref{thm:privacy-FO} and \ref{thm:error-FO} (for the frequency oracle) and Theorems~\ref{thm:privacy-sh-promise} and \ref{thm:error-sh-promise} (for the succinct histogram under the promise).

In the promise problem, the main advantage was the lack of interference from the users who do not hold the heavy hitter $v^*$ in question. The main idea here is to obtain a reduction in which we create the conditions of the promise problem separately for each heavy hitter $v^*\in\V$ such that the extra computational cost is at most a small $\poly(n)$ factor. To do this, we hash each item $v\in\V$ into one of $K$ separate \emph{parallel} channels such that users holding the same item will transmit their reports in the same channel. Each user, in the remaining $K-1$ channels, will simulate an ``idle'' user with item $\bot$ as in the promise problem. By choosing $K$ sufficiently large, and repeating the protocol in \emph{parallel} for $T$ times\footnote{That is, the total number of parallel channels is $KT$. In each group of $K$ channels, a fresh hash seed is used.}, we can guarantee that, with high probability, \emph{every} heavy hitter $v^*\in\V$ gets assigned to an interference-free channel. Hence, by using an error-optimal construction for the promise problem like $\sh_{\sf PP}$ in each one of these channels, we eventually obtain a list of at most $KT$ items such that, with high probability, all the heavy hitters will be on that list. However, this list may also contain other erroneously decoded items due to hash collisions and we do not know which items on the list are the heavy hitters. To overcome this, in a separate parallel channel of the protocol, we run a frequency oracle protocol (like $\prot$-$\fo$) and use the resulting frequency oracle to estimate the frequencies of all the items on that list, then output all the items whose estimated frequencies are above $\eta$ together with their estimated frequencies.

For the purpose of this construction, it suffices to use a pairwise independent hash function. A family of functions $\hc=\left\{h_{s} \vert h_{s}:\V\rightarrow [K], s\in \{0,1\}^{\ell}\right\}$ is said to be pairwise independent if for any distinct pair $v\neq v'\in\V$ and any values $j, k \in [K]$, a uniformly sampled member of such a family $h_s, s\leftarrow \{0, 1\}^{\ell}$, satisfies both $h_s(v)=j$ \emph{and} $h_s(v')=k$ simultaneously with probability $\frac{1}{K^2}$. There are efficient constructions of pairwise independent hash families with seed length $\ell=O\left(\max\left(\log(d), \log(K)\right)\right)$. In our construction, we can use any instance of such a family as long as it is efficient. Our hash family (or simply hash) is denoted by $\hash$ that, for a given input seed $s$ and an item $v\in\V$, returns a number in $[K]$.  All users and the server are assumed to have access to $\hash$.  Moreover, we use a source of \emph{public} randomness $\gen$ that, on an input integer $\ell>0$, generates a random uniform string from $\{0,1\}^{\ell}$ that is seen by everyone\footnote{We may also think of $\gen$ as being run at the server which then announces the resulting random string to all the users.}.

The parameters of our hash family are $K=n^{3/2}$ and $\ell=O\left(\max\left(\log(d), \log(n)\right)\right)$. Our construction protocol~$\prot$-$\sh$ is given by Algorithm~\ref{Alg:SH-eff} below.

\fi

\begin{algorithm}[htb]
	\caption{$\prot$-$\sh$: $\eps$-$\ldp$ Efficient Protocol for Succinct Histograms}
	\begin{algorithmic}[1]
		\REQUIRE Users' inputs $\{v_i\in\V\cup\{\bot\}: i\in[n]\}$, the privacy parameter $\eps$, and the confidence parameter $\beta>0$.
	\STATE $\List\leftarrow \emptyset~$ (initialize list of heavy hitters to the empty set.)
	\STATE $\ell\leftarrow 2\max\left(\log(d), \log(n)\right)~;~$ $K\leftarrow n^{3/2}$.
	\STATE $T\leftarrow \left\lceil\log(3/\beta)\right\rceil$\label{set_T}
	\FOR{ $t=1$ to $T$}
	{\STATE $s_t\leftarrow\gen(\ell)$.}
		\FOR{ Channels $k=1$ to $K$ }
			{\FOR{ Users $i=1$ to $n$ }
			{\STATE If $\hash(s_t, v_i)\neq k$, set $v'_i\leftarrow \bot$. Else, set $v'_i\leftarrow v_i$}
			\ENDFOR}
		{\STATE $\hat{v}\leftarrow$ $\prot^{\sf PP}$-$\sh_{\sf
                    PP}\left(\left\{v'_1,..., v'_n\right\};
                    \frac{\eps}{2T+1}; \frac{\beta}{3d}\right)$
                    \COMMENT{i.e.,
                  run $\prot^{\sf PP}$-$\sh_{\sf PP}$ on the modified set of users'
                  items 
                        to obtain an estimate $\hat{v}$ of the possibly unique item
                        transmitted in the $k$-th channel.}
                      }\label{step-eff-prot}
	   {\STATE If $\hat{v}\notin\List$, then add $\hat{v}$ to $\List$.}
	  \ENDFOR
        \ENDFOR
	\STATE $\fo\leftarrow$ $\prot$-$\fo\left(\left\{v_1, ..., v_n\right\};~
            \frac{\eps}{2T+1};~ \beta/3\right)$
            \COMMENT{i.e., run $\prot$-$\fo$ on the original set of
              users' items to obtain the frequency oracle $\fo$.
                                }\label{step-fo-eff}
	\FOR{ $\hat{v} \in \List$ }
	{\STATE $\hf(\hat{v})\leftarrow \A_{{\sf FO}}\left(\fo, \hat{v}\right)$.\label{step-freq} \COMMENT{$\A_{{\sf FO}}$ is the frequency estimator given in Section~\ref{subsec:JL}.}}

	{\STATE If $\hf(\hat{v})< \frac{2T+1}{\eps}\sqrt{\frac{\log(d)\log(1/\beta)}{n}}$,~ remove $\hat{v}$ from $\List$.}
	\ENDFOR
        \RETURN $ \Big\{ \big(v , \hf\left(v\right)\big) \ :\  v\in\List\Big\}.$
	\end{algorithmic}
	\label{Alg:SH-eff}
\end{algorithm}

It is not hard to see that the \emph{total} computational cost of this construction is
$$O\left(n^{3/2}\log(1/\beta)\cost_{\sf PP} + \cost_{\sf FO}+n \cost_{\A_{\sf FO}}\right)$$
where $\cost_{\sf PP}$, $\cost_{\sf FO}$, and $\cost_{\A_{\sf FO}}$ are the computational costs of the promise problem sub-protocol, the frequency oracle sub-protocol, and the algorithm that computes a given frequency estimate, respectively.  Hence, for our choice of the sub-protocols above, one can easily verify the overall worst case cost of our construction $O(n^{5/2}\poly(\log(d))\log(1/\beta))$.

The report length of each user is now scaled by $KT$ compared to that of the promise problem, that is, $O\left(n^{3/2}\log(1/\beta)\log\left(\log(d)\right)\right)$.
In the next section, we will discuss an approach that gets it down to $1$ bit at the expense of increasing the public coins.

\ifnum\full=1
Our construction here relies on public randomness represented by the $T$ fresh random strings (seeds) of each of length $2\max\left(\log(d), \log(n)\right)$ which for the setting we consider\footnote{We assume $d\gg n$ for our definitions of computational efficiency and succinctness to be meaningful.} is $O(\log(d))$. Hence, the total number of public coins needed is $O\left(\log(1/\beta)\log(d)\right)$.

\begin{thm}[Privacy of $\prot$-$\sh$]
Protocol $\prot$-$\sh$ given by Algorithm~\ref{Alg:SH-eff} is $\eps$-differentially private.
\end{thm}\label{thm:privacy-eff-prot}

\else

\begin{thm} 
Protocol $\prot$-$\sh$ given by Algorithm~\ref{Alg:SH-eff} is $\eps$-differentially private.
\end{thm}\label{thm:privacy-eff-prot}

\fi

\begin{proof}
First, observe that Protocol $\prot$-$\sh$ runs Protocol $\prot^{\sf PP}$-$\sh_{\sf PP}$ over $KT$ channels and runs Protocol $\prot$-$\fo$ once over a separate channel. In the first $KT$ channels, for any fixed sequence of the values of the seed of the hash function, the reports of each user over these channels are independent. Moreover, each user gets assigned to exactly $T$ channels. Fix any user $i$ and any two items $v_i, v'_i\in\V$. Using these observations, one can see that, for any fixed sequence of values of the seed of the hash over these $KT$ channels, the distribution of the report of user $i$ when its item is $v_i$ differs from the distribution when the user's item is $v'_i$ in \emph{at most} $2T$ channels, and in each of these channels, the ratio between the two distributions is at most $e^{\frac{\eps}{2T+1}}$ by the differential privacy of $\prot^{\sf PP}$-$\sh_{\sf PP}$ (note that the input privacy parameter to in Step~\ref{step-eff-prot} is $e^{\frac{\eps}{2T+1}}$). Hence, by independence of the user's reports over separate channels, the corresponding ratio over all the $KT$ channels is at most $e^{\frac{2T\eps}{2T+1}}$. In the separate channel for the frequency oracle protocol, again by the differential privacy of $\prot$-$\fo$, this ratio is bounded by $e^{\frac{\eps}{2T+1}}$. Putting this together with the argument in the previous paragraph completes the proof.
\end{proof}

\ifnum\full=1

\begin{thm}[Error of $\prot$-$\sh$]
For any set of users' items $\{v_1, ..., v_n\}$ and any $\beta>0$, there is a number $\eta=O\left(\frac{\log^{\frac{3}{2}}(1/\beta)}{\eps}\sqrt{\frac{\log(d)}{n}}\right)$ such that, with probability at least $1-\beta$, Protocol $\prot$-$\sh$ outputs $\left\{ \left(v , \hf\left(v\right)\right) \ : \  v\in\List\right\}$ where  $\List=\{v^*\in\V: f(v^*)\geq \eta\}$ (i.e., a list of all items whose frequencies are greater than $\eta$), and the error in the frequency estimates satisfies
$$\max\limits_{v\in\V}\vert\hf(v)-f(v)\vert=O\left(\frac{\log^{\frac{3}{2}}(1/\beta)}{\eps}\sqrt{\frac{\log(d)}{n}}\right).$$
(As mentioned before, the frequency estimates of items $v\notin\List$ are implicitly zero.)
\end{thm}\label{thm:error-eff-prot}\rbnote{changed dependence on $\beta$ in all $\eta$ and error expressions from $\log(d/\beta)$ to $\log(d)\log(1/beta)$ in the theorem and the proof.}

\else

\begin{thm} 
For any set of users' items $\{v_1, ..., v_n\}$ and any $\beta>0$, there is a number $\eta=O\left(\frac{\log^{\frac{3}{2}}(1/\beta)}{\eps}\sqrt{\frac{\log(d)}{n}}\right)$ such that, with probability at least $1-\beta$, Protocol $\prot$-$\sh$ outputs $\left\{ \left(v , \hf\left(v\right)\right) \ : \  v\in\List\right\}$ where  $\List=\{v^*\in\V: f(v^*)\geq \eta\}$ (i.e., a list of all items whose frequencies are greater than $\eta$), and the error in the frequency estimates satisfies
$$\max\limits_{v\in\V}\vert\hf(v)-f(v)\vert=O\left(\frac{\log^{\frac{3}{2}}(1/\beta)}{\eps}\sqrt{\frac{\log(d)}{n}}\right).$$
(As mentioned before, the frequency estimates of items $v\notin\List$ are implicitly zero.)
\end{thm}\label{thm:error-eff-prot}


\fi

\begin{proof}
Let $\U$ denote the set of the users' items $\{v_1, ..., v_n\}$. We first show that for the setting of $K$ and $T$ in Algorithm~\ref{Alg:SH-eff}, running $\prot^{\sf PP}$-$\sh_{\sf PP}$ over $KT$ channels will isolate every heavy hitter (i.e., every item occurring with frequency at least $\eta$) in at least one channel without interference from other items. Let $\He=\{v^*\in\V:~ f(v^*)\geq\eta\}$ denote the set of the heavy hitters. Note that $\vert\He\vert\leq\frac{1}{\eta}$.
\begin{claim}
If $\eta\geq\frac{2T+1}{\eps}\sqrt{\frac{\log(d)\log(1/\beta)}{n}}$, then with probability at least $1-\beta/3$ (over the sequence of the seed values $s_1, ..., s_T$ of $\hash$), for every heavy hitter $v^*\in\He$ there is $t\in[T]$ such that $\hash(s_t, v^*)\neq\hash(s_t, v)$ for all $v\in\U\setminus\{v^*\}$.
\end{claim}

First, we prove this claim. Fix $v^*\in\He$. Let $t\in[T]$. Let $\coll_{s_t}(v^*)\triangleq\vert\{i\in[n]:~\hash(s, v^*)=\hash(s, v_i), v_i\neq v^*\}\vert$ denote the number of collisions between $v^*$ and users' items that are different from $v^*$ when the hash seed is $s_t$. First, we bound the expected number of such collisions:
\begin{align}
\E[\coll_{s_t}(v^*)]&\leq\sum\limits_{i:v_i\neq v^*}\frac{1}{K}\leq \frac{n}{K}=\frac{1}{\sqrt{n}}\nonumber
\end{align}
Hence, by Markov's inequality, with probability at least $1-\frac{1}{\sqrt{n}}$, $\coll_{s_t}(v^*)=0$. Hence, with probability at least $1-\frac{1}{\eta}\left(\frac{1}{\sqrt{n}}\right)^T\geq 1-\beta/3$, for each $v^*\in\He$, there exists $t\in[T]$ such that $\coll_{s_t}(v^*)=0$, which proves the claim.

\vspace{0.2cm}

This implies that with probability at least $1-\beta/3$, there is a set $\W\subset[KT]$ of ``good'' channels whose size is the same as the number of heavy hitters such that each heavy hitter $v^*\in\He$ is hashed into one of these channels without collisions. Conditioned on this event, let $w\in\W$ and let $v_w^*$ denote the heavy hitter in channel $w$. By Theorem~\ref{thm:error-sh-promise}, running Protocol $\prot^{\sf PP}$-$\sh_{\sf PP}$ over channel $w$ yields the correct estimate of $v_w^*$ with probability at least $1-\frac{\beta}{3d}$ (Step~\ref{step-eff-prot} of Algorithm~\ref{Alg:SH-eff}). Hence, with probability at least $1-\beta/3$, all estimates of $\prot^{\sf PP}$-$\sh_{\sf PP}$ in all channels in $\W$ are correct. Hence, at this point, with probability at least $1-\frac{2\beta}{3}$, $\List$ contains all the heavy hitters in $\He$ among other possibly unreliable estimates of $\prot^{\sf PP}$-$\sh_{\sf PP}$ for the channels in $[KT]\setminus\W$.

Now, conditioned on the event above, by the error guarantee of $\fo$ given by Theorem~\ref{thm:error-FO}, with probability at least $1-\beta/3$, the maximum error in the frequency estimates of all the items in $\List$ (Step~ \ref{step-freq} of Algorithm~\ref{Alg:SH-eff}), denoted by $\err\left(\List\right)$, is 
$$O\left(\frac{2T+1}{\eps}\sqrt{\frac{\log(d/\beta)}{n}}\right)=O\left(\frac{\log(1/\beta)}{\eps}\sqrt{\frac{\log(d/\beta)}{n}}\right).$$
Hence, all those items in $\List$ with actual frequencies greater than 
\ifnum\full=1
$$\eta\triangleq \frac{2T+1}{\eps}\sqrt{\frac{\log(d)\log(1/\beta)}{n}} + \err\left(\List\right)=O\left(\frac{\log^{\frac{3}{2}}(1/\beta)}{\eps}\sqrt{\frac{\log(d)}{n}}\right)$$ 
\else
\begin{align}
\eta&\triangleq \frac{2T+1}{\eps}\sqrt{\frac{\log(d)\log(1/\beta)}{n}} + \err\left(\List\right)\nonumber\\
&=O\left(\frac{\log^{\frac{3}{2}}(1/\beta)}{\eps}\sqrt{\frac{\log(d)}{n}}\right)\nonumber
\end{align}
\fi
will be kept in $\List$ whereas all those items with frequency less than $\eta$ will be removed. Note that the frequency estimates that are implicitly assumed to be zero of those items that are not on the list cannot have error greater than $\eta$ since their actual frequencies are less than $\eta$. This completes the proof.
\end{proof}


\ifnum\full=1

\section{The Full Protocol}\label{sec:full-protocol}

\subsection{Generic Protocol with $1$-Bit Reports}\label{subsec:generic-1bit}

In this section, we give a generic approach that transforms \emph{any}
private protocol in the distributed setting (not necessarily for
frequency estimation or succinct histograms) to a private distributed
protocol where the report of each user is a single bit at the expense
of adding to the overall original public randomness a number of bits
that is $O(n\tau)$ where $\tau$ is the length of each user's report in the
original protocol. As mentioned in the introduction, the
transformation is a modification of the general compression
technique of McGregor et al. \cite{MMPRTV10}.

Consider a generic private distributed protocol $\gprot$ in which $n$ users are communicating with an untrusted server. For any $n\in\mathbb{N}$, the protocol follows the following general steps. As before, each user $i\in[v]$ has a data point $v_i$ that lives in some finite set $\V=[d]$. Let $\Q_i:\V\cup\{\bot\}\rightarrow \mathcal{Z}$ be \emph{any} $\eps$-local randomizer of user $i\in[n]$. We assume, w.l.o.g., that $\Q_i$ may also take a special symbol $\bot$ as an input. Each user runs its $\eps$-local randomizer $\Q_i$ on its input data $v_i$ (and any public randomness in the protocol, if any) and outputs a report $\bz_i$. For simplicity, each report $\bz_i$ is assumed to be a binary string of length $\tau$. Let $\stat: \V^n\rightarrow \C$ be some statistic that the server wishes to estimate where $\C$ is some bounded subset of $\re^k$ for some integer $k>0$. The server collects the reports $\{\bz_i: i\in[n]\}$ and runs some algorithm $\A_{\stat}$ on the users' reports (and the public randomness) and outputs an estimate $\widehat{\stat}\in\C$ of $\stat\left(v_1, ..., v_n\right)$.

We now give a generic construction $1$-{\sf Bit}-$\gprot$ for protocol $\gprot$ where each user's report is one bit (See Algorithm~\ref{Alg:1bit-gen} below).

\begin{algorithm}[htb]
	\caption{$1$-{\sf Bit}-$\prot$: $\eps$-$\ldp$ Generic 1-Bit Protocol}
	\begin{algorithmic}[1]
		\REQUIRE Users' inputs $\{v_i\in\V: i\in[n]\}$ and a privacy parameter $\eps \leq \ln(2)$.
        \STATE  Generate $n$ independent public strings $y_1\leftarrow \Q_1(\bot), ..., y_n\leftarrow \Q_n(\bot)$.
       \FOR{ Users $i=1$ to $n$}
            {\STATE Compute $p_i=\frac{1}{2}\frac{\Pr\left[\Q_i\left(v_i\right)=y_i\right]}{\Pr\left[\Q_i(\bot)=y_i\right]}$.}\label{step:prob-comp}
	{\STATE Sample a bit $b_i$ from $\mbox{Bernoulli}(p_i)$ and sends it to the server.}
        \ENDFOR
	\STATE $\mbox{Reports}\leftarrow \emptyset.$ \COMMENT{Server initialize the set of collected reports.}
        \FOR{ $i=1$ to $n$}
            {\STATE Server checks if $b_i=1$, add $y_i$ to $\mbox{Reports}$.}\label{rep-gen}
        \ENDFOR
	\STATE  $\widehat{\stat}\leftarrow\A_{\stat}\left(\mbox{Reports}\right).$ \COMMENT{Run algorithm $\A_{\stat}$ on the collected reports to obtain an estimate of the desired statistic as described in the original protocol $\gprot$.}
	\RETURN $\widehat{\stat}.$
	\end{algorithmic}
	\label{Alg:1bit-gen}
\end{algorithm}

Note also that the only additional computational cost in this generic transformation is in Step \ref{step:prob-comp}. If computing these probabilities can be done efficiently, then this transformation preserves the computational efficiency of the original protocol.

\begin{thm}[Privacy of $1$-{\sf Bit}-$\prot$]
Protocol $1$-{\sf Bit}-$\prot$ given by Algorithm~\ref{Alg:1bit-gen} is $\eps$-$\ldp$.
\end{thm}\label{thm:privacy-1bit-gen}
\begin{proof}
Consider the output bit $b_i$ of any user $i\in[n]$. First, note that $p_i$ (in Step \ref{step:prob-comp}) is a valid probability since for any item $v_i\in\V$,  the right-hand side of Step~\ref{step:prob-comp} is at most $\frac{e^{\eps}}{2}$ by $\eps$-differential privacy of $\Q_i$, and since $\eps\leq \ln(2)$, $p_i\leq 1$. For any $v\in\V$ and any public string $y_i$, let $p_i(v, y_i)$ denote the conditional probability that $b_i=1$ given that $\Q_i(\bot)=y_i$ when the item of user $i$ is $v$. Let $v, v'\in\V$ be any two items. It is easy to see that $\frac{p_i(v, y_i)}{p_i(v', y_i)}=\frac{\Pr[\Q_i(v)=y_i]}{\Pr[\Q_i(v')=y_i]}$ which lies in $[e^{-\eps}, e^{\eps}]$ by $\eps$-differential privacy of $\Q_i$. One can also verify that $\frac{1-p_i(v, y_i)}{1-p_i(v', y_i)}$ also lies in $[e^{-\eps}, e^{\eps}]$.
\end{proof}

One important feature in the construction above is that the conditional distribution of the public string $y_i$ given that $b_i=1$ is exactly the same as the distribution of $\Q_i(v_i)$, i.e., $\Pr\left[\Q_i(\bot)=y_i \vert~b_i=1\right]$ $=\Pr\left[\Q_i(v_i)=y_i\right]$, and hence, upon receiving a bit $b_i=1$ from user $i$, the server's view of $y_i$ is the same as its view of an actual report $\bz_i\leftarrow\Q_i(v_i)$ as it was the case in the original protocol.

We note that the probability that a user $i\in[n]$ accepts (sets $b_i=1$) taken over the randomness of $y_i$ is
$$\frac{1}{2}\sum_{y}\frac{\Pr\left[\Q_i\left(v_i\right)=y\right]}{\Pr\left[\Q_i(\bot)=y\right]}\cdot\Pr\left[\Q_i(\bot)=y\right]=\frac{1}{2}.$$

\mypar{Key statement} The two facts above show that our protocol is functionally equivalent to: first, sampling a subset of the users where each user is sampled independently with probability $1/2$, then  running the original protocol $\gprot$ on the sample. Thus, if the original protocol is resilient to sampling, meaning that its error performance (with respect to some notion of error) is not essentially affected by this sampling step, then the generic transformation given by Algorithm~\ref{Alg:1bit-gen} will have essentially the same error performance.

We now formalize this statement. Let $\psi:\C\times\C\rightarrow [0, \infty]$ be some notion of error (not necessarily a metric) between any two points in $\C$. For any given set of users' data $\{v_1, ..., v_n\}$, the error of the protocol $\gprot$ is defined as
\begin{align}
\Ec_{\psi}\left(\stat; \gprot\left(v_1, ..., v_n\right)\right)&\triangleq \psi\left(\stat(v_1, ..., v_n), ~\widehat{\stat}\right)\label{gen-error}
\end{align}

Let $\samp$ be a random sampling procedure that takes any set of users' data $\{v_1, ..., v_n\}$ and constructs a set $\Sc$ by sampling each point $v_i, i\in[n]$ independently with probability $1/2$. We say that $\gprot$ is \emph{sampling-resilient} in estimating $\stat$ with respect to $\psi$ if for any set of users' data $\left(v_1, ..., v_n\right)\in\V^n$ and any $\beta>0$, whenever
$$\Ec_{\psi}\left(\stat; \gprot\left(v_1, ..., v_n\right)\right)=O\left(g\left(n, d, k, \eps\right)\right)$$
for some non-negative function $g$ with probability at least $1-\beta$ over all the randomness in $\gprot$, then
$$\Ec_{\psi}\left(\stat; \gprot\left(\samp\left(v_1, ..., v_n\right)\right)\right)=O\left(g\left(n, d, k, \eps\right)\right)$$
with probability at least $1-2\beta$ over all the randomness in $\gprot$ and $\samp$.


\begin{thm}[Characterization of error under sampling-resilience]
Suppose that for any set of users' data and any $\beta>0$, the original protocol $\gprot$ has error (with respect to $\psi$) that is bounded by some non-negative number $g=g(n, d, k, \eps)$ with probability at least $1-\beta$ over the randomness in $\gprot$. If $\gprot$ is sampling-resilient in estimating $\stat$ with respect to $\psi$, then construction $1$-{\sf Bit}-$\prot$ yields error (with respect to $\psi$) that is $O\left(g(n, d, k, \eps)\right)$.
\end{thm}\label{thm:error-1bit-gen}

\subsection{Efficient Construction of Succinct Histograms with $1$-Bit Reports}\label{subsec:eff-protocol-1bit}

We now apply the generic transformation discussed above to our efficient protocol $\prot$-$\sh$ (Algorithm~\ref{Alg:SH-eff}) to obtain an \emph{efficient} private protocol for succinct histograms with $1$-bit reports and optimal error.

\mypar{Computational efficiency} To show that the protocol remains efficient after this transformation, we argue that the probabilities in Step~\ref{step:prob-comp} of Algorithm~\ref{Alg:1bit-gen} can be computed efficiently in our case. The overall $\eps$-local randomizer $\Q^{\sf Full}_i$ at each user $i$ over all the $KT+1$ parallel channels in $\prot$-$\sh$ is described in Algorithm~\ref{Alg:overall-loc-rand}. Note that given the user's item $v_i$ and the seed of the hash, the $KT+1$ components of $\Q^{\sf Full}_i(v_i)$ are independent. Moreover, note that $(K-1)T$ of these components have the same (uniform) distribution since each user gets assigned by the hash function to only $T+1$ channels and in the remainder channels the user's report is uniformly random.  Hence, to execute Step~\ref{step:prob-comp} of Algorithm~\ref{Alg:1bit-gen}, each user in our case only needs to compute $T+1$ probabilities out of the total $KT+1$ components. This is easy because of the way the basic randomizer $\R$ works. To see this, first note that each $y_i$ (referring to the public string $y_i$ in Algorithm~\ref{Alg:1bit-gen}) is now a sequence of $(\indx, \bit)$ pairs: $\left(j_1, b_{j_1}\right), \ldots, \left(j_{KT+1}, b_{j_{KT+1}}\right)$. To compute the probability corresponding to one of the $T+1$ item-dependent components of $\Q^{\sf Full}_i(v_i)$, each user first locates in the public string $y_i$ the pair $(j, b)$ corresponding to this component. Then, it compares the sign of the $j$-th bit of the encoding\footnote{This encoding is either $\bc(v_i)$ or $\phi_{v_i}$ depending on whether we are at Step~\ref{step:enc-ecc} or Step~\ref{step:enc-matrix} of Algorithm~\ref{Alg:overall-loc-rand}.} of its item $v_i$ with the sign of $b$.  If signs are equal, then the desired probability is $\frac{e^{\eps}}{1+e^{\eps}}$, otherwise it is $\frac{1}{1+e^{\eps}}$. Hence, the computational cost of this step (per user) is $O\left(T\log\left(m_{\sf PP}\right)+\log\left(m_{\sf FO}\right)\right)=O(\log\left(\log(d)\right)\log(1/\beta)+\log(n))$ where $m_{\sf PP}$ is the length of the encoding $\bc(v_i)$ used in the promise problem protocol $\prot^{\sf PP}$-$\sh_{\sf PP}$ and $m_{\sf FO}$ is the length of the encoding $\phi_{v_i}$ used in the frequency oracle protocol $\prot$-$\fo$. Thus, at worst the overall computational cost of the $1$-Bit protocol is the same as that of protocol $\prot$-$\sh$.

\begin{algorithm}[htb]
	\caption{$\Q^{\sf Full}_i$: $\eps$-Local Randomizer of User $i$ in $\prot$-$\sh$ (Algorithm\ref{Alg:SH-eff})}
	\begin{algorithmic}[1]
		\REQUIRE item $v_i\in\V$, privacy parameter $\eps$, seeds of $\hash$~ $s_1, ..., s_T$.
        \FOR{ $t=1$ to $T$}
		\FOR{ Channels $k=1$ to $K$ }
            {\STATE If $\hash(s_t, v_i)\neq k$, set $\bz^{(t,k)}_i=\R_i\left(\zeroB, \eps\right)$. Else, set $\bz^{(t,k)}_i=\R_i\left(\bc\left(v_i\right), \eps\right)$. \COMMENT{ $\bz^{(t,k)}_i$ denotes the report of user $i$ in the $k$-th channel in the $t$-th group.}}\label{step:enc-ecc}
        \ENDFOR
        \ENDFOR
   \STATE Set $\bz^{(fo)}_i=\R_i\left(\phi_{v_i}, \eps\right)$. \COMMENT{$\phi_{v_i}$ is the $v_i$-th column of $\Phi$ the encoding matrix in the construction of the frequency oracle $\fo$.}\label{step:enc-matrix}
\RETURN $\bz_i=\left(\bz^{(t,k)}_i,~ \bz^{(fo)}_i~:~ t=1, ..., T; k=1, ..., K\right)$.
\end{algorithmic}
	\label{Alg:overall-loc-rand}
\end{algorithm}

Our $1$-Bit Protocol gives the same privacy and error guarantees of $\prot$-$\sh$.

\begin{thm}[Privacy of the $1$-Bit Protocol for Succinct Histograms]
The $1$-Bit Protocol for succinct histograms is $\eps$-differentially private.
 \end{thm}\label{thm:privacy-1bit-prot-sh}
\begin{proof}
The proof follows directly from Theorems~\ref{thm:privacy-eff-prot} and \ref{thm:privacy-1bit-gen}.
\end{proof}

\begin{thm}[Error of the $1$-Bit Protocol for Succinct Histograms]
The $1$-Bit Protocol for succinct histograms provides the same guarantees of Protocol $\prot$-$\sh$ given in Theorem~\ref{thm:error-eff-prot}.
\end{thm}\label{thm:error-1bit-prot-sh}
\begin{proof}
The proof follows from Theorem~\ref{thm:error-1bit-gen} and the fact that Protocol $\prot$-$\sh$ is sampling-resilient. Note that for any $\beta>0$, \emph{any} $n$, and \emph{any} set of users' items $\{v_1, ..., v_n\}$,  $\prot$-$\sh$ satisfies the error guarantees of Theorem~ \ref{thm:error-1bit-gen} with probability at least $1-\beta$. Thus, sampling a subset of users' items (where each item is picked with probability $1/2$), then feeding this set to $\prot$-$\sh$ will only lead to an extra error term of $O(1/\sqrt{n})$ which will be swamped by the original error term of $O\left(\sqrt{\frac{\log(d)}{n}}\right)$.
\end{proof}


\else 

\section{The Full Protocol}\label{sec:full-protocol}

\subsection{Generic Protocol with $1$-Bit Reports}\label{subsec:generic-1bit}

In this section, we give a generic approach that transforms \emph{any}
private protocol in the distributed setting (not necessarily for
frequency estimation or succinct histograms) to a private distributed
protocol where the report of each user is a single bit at the expense
of adding to the overall original public randomness a number of bits
that is $O(n\tau)$ where $\tau$ is the length of each user's report in the
original protocol. As mentioned in the introduction, the
transformation is a modification of the general compression
technique of McGregor et al. \cite{MMPRTV10}.

Let $\gprot$ be any private distributed protocol in which $n$ users are communicating with an untrusted server. $\gprot$ follows the following general steps. Each user $i\in[v]$ has a data point $v_i\in\V=[d]$. Let $\Q_i:\V\cup\{\bot\}\rightarrow \mathcal{Z}$ be \emph{any} $\eps$-local randomizer used by user $i\in[n]$. We assume, w.l.o.g., that $\Q_i$ may also take a special symbol $\bot$ as an input. Each user runs its $\eps$-local randomizer $\Q_i$ on its data $v_i$ (and any public randomness in the protocol, if any) and outputs a report $\bz_i$. For simplicity, each report $\bz_i$ is assumed to be a binary string of length $\tau$. Let $\stat: \V^n\rightarrow \C$ be some statistic about the data that the server wishes to estimate where $\C$ is some bounded subset of $\re^k$ for some integer $k>0$. The server collects the reports $\{\bz_i: i\in[n]\}$ and runs some algorithm $\A_{\stat}$ on the reports (and the public randomness) and outputs an estimate $\widehat{\stat}\in\C$ of $\stat\left(v_1, ..., v_n\right)$.

We now give a generic construction $1$-{\sf Bit}-$\gprot$ (Algorithm \ref{Alg:1bit-gen}) for $\gprot$ where each user's report is one bit.

\begin{algorithm}[htb]
	\caption{$1$-{\sf Bit}-$\prot$: $\eps$-$\ldp$ Generic 1-Bit Protocol}
	\begin{algorithmic}[1]
		\REQUIRE Users' inputs $\{v_i\in\V: i\in[n]\}$ and a privacy parameter $\eps \leq \ln(2)$.
        \STATE  Generate $n$ independent public strings $y_1\leftarrow \Q_1(\bot), ..., y_n\leftarrow \Q_n(\bot)$.
       \FOR{ Users $i=1$ to $n$}
            {\STATE Compute $p_i=\frac{1}{2}\frac{\Pr\left[\Q_i\left(v_i\right)=y_i\right]}{\Pr\left[\Q_i(\bot)=y_i\right]}$.}\label{step:prob-comp}
	{\STATE Sample a bit $b_i$ from $\mbox{Bernoulli}(p_i)$ and sends it to the server.}
        \ENDFOR
	\STATE $\mbox{Reports}\leftarrow \emptyset.$ \COMMENT{Server initialize the set of collected reports.}
        \FOR{ $i=1$ to $n$}
            {\STATE Server checks if $b_i=1$, add $y_i$ to $\mbox{Reports}$.}\label{rep-gen}
        \ENDFOR
	\STATE  $\widehat{\stat}\leftarrow\A_{\stat}\left(\mbox{Reports}\right).$ \COMMENT{Run algorithm $\A_{\stat}$ on the collected reports to obtain an estimate of the desired statistic as described in the original protocol $\gprot$.}
	\RETURN $\widehat{\stat}.$
	\end{algorithmic}
	\label{Alg:1bit-gen}
\end{algorithm}

Note that the only additional computational cost in this generic transformation is in Step \ref{step:prob-comp}. If computing these probabilities can be done efficiently, then this transformation preserves the computational efficiency of the original protocol.

\begin{thm}
\label{thm:privacy-1bit-gen} 
Protocol $1$-{\sf Bit}-$\prot$ given by Algorithm~\ref{Alg:1bit-gen} is $\eps$-$\ldp$.
\end{thm}
\begin{proof}
Consider the output bit $b_i$ of any user $i\in[n]$. First, note that $p_i$ (in Step \ref{step:prob-comp}) is a valid probability since for any item $v_i\in\V$, the right-hand side of Step~\ref{step:prob-comp} is at most $\frac{e^{\eps}}{2}$ by $\eps$-differential privacy of $\Q_i$, and since $\eps\leq \ln(2)$, $p_i\leq 1$. For any $v\in\V$ and any public string $y_i$, let $p_i(v, y_i)$ denote the conditional probability that $b_i=1$ given that $\Q_i(\bot)=y_i$ when the item of user $i$ is $v$. Let $v, v'\in\V$ be any two items. It is easy to see that $\frac{p_i(v, y_i)}{p_i(v', y_i)}=\frac{\Pr[\Q_i(v)=y_i]}{\Pr[\Q_i(v')=y_i]}$ which lies in $[e^{-\eps}, e^{\eps}]$ by $\eps$-differential privacy of $\Q_i$. One can also verify that $\frac{1-p_i(v, y_i)}{1-p_i(v', y_i)}$ also lies in $[e^{-\eps}, e^{\eps}]$.
\end{proof}

One important feature in the construction above is that the conditional distribution of the public string $y_i$ given that $b_i=1$ is exactly the same as the distribution of $\Q_i(v_i)$, and hence, upon receiving a bit $b_i=1$ from user $i$, the server's view of $y_i$ is the same as its view of an actual report $\bz_i\leftarrow\Q_i(v_i)$.

Note also that the probability that a user $i\in[n]$ accepts (sets $b_i=1$) taken over the randomness of $y_i$ is
$$\frac{1}{2}\sum_{y}\frac{\Pr\left[\Q_i\left(v_i\right)=y\right]}{\Pr\left[\Q_i(\bot)=y\right]}\cdot\Pr\left[\Q_i(\bot)=y\right]=\frac{1}{2}.$$

\boldpar{Key statement} \emph{The two facts above show that our protocol is functionally equivalent to: first, sampling a subset of the users where each user is sampled independently with probability $1/2$, then,  running the original protocol $\gprot$ on the sample. Thus, if the original protocol is resilient to sampling, i.e., its error performance (with respect to some notion of error) is not essentially affected by this sampling step, then the generic transformation given by Algorithm~\ref{Alg:1bit-gen} will have essentially the same error performance.}

\subsection{Efficient Construction with $1$-Bit Reports}\label{subsec:eff-protocol-1bit}

We now apply the generic transformation discussed above to our efficient protocol $\prot$-$\sh$ (Algorithm~\ref{Alg:SH-eff}) to obtain an \emph{efficient} private protocol for succinct histograms with $1$-bit reports and optimal error. The fact that such protocol has the same error as $\prot$-$\sh$ follows from the key statement above and the fact that $\prot$-$\sh$ is resilient to sampling. 

\begin{thm}
\label{thm:privacy-1bit-prot-sh}
The $1$-Bit Protocol for succinct histograms is $\eps$-differentially private.
 \end{thm}

\begin{thm}
\label{thm:error-1bit-prot-sh}
The $1$-Bit Protocol for succinct histograms provides the same guarantees of Protocol $\prot$-$\sh$ given in Theorem~\ref{thm:error-eff-prot}.
\end{thm}

\mypar{Computational efficiency} To show that the protocol remains efficient after this transformation, we argue that the probabilities in Step~\ref{step:prob-comp} of Algorithm~\ref{Alg:1bit-gen} can be computed efficiently in our case. The overall $\eps$-local randomizer $\Q^{\sf Full}_i$ at each user $i$ over all the $KT+1$ parallel channels in $\prot$-$\sh$ is described in Algorithm~\ref{Alg:overall-loc-rand}. Note that given the user's item $v_i$ and the seed of the hash, the $KT+1$ components of $\Q^{\sf Full}_i(v_i)$ are independent. Moreover, note that $(K-1)T$ of these components have the same (uniform) distribution since each user gets assigned by the hash function to only $T+1$ channels and in the remainder channels the user's report is uniformly random.  Hence, to execute Step~\ref{step:prob-comp} of Algorithm~\ref{Alg:1bit-gen}, each user in our case only needs to compute $T+1$ probabilities out of the total $KT+1$ components. This is easy because of the way the basic randomizer $\R$ works. To see this, first note that each $y_i$ (referring to the public string $y_i$ in Algorithm~\ref{Alg:1bit-gen}) is now a sequence of $(\indx, \bit)$ pairs: $\left(j_1, b_{j_1}\right), \ldots, \left(j_{KT+1}, b_{j_{KT+1}}\right)$. To compute the probability corresponding to one of the $T+1$ item-dependent components of $\Q^{\sf Full}_i(v_i)$, each user first locates in the public string $y_i$ the pair $(j, b)$ corresponding to this component. Then, it compares the sign of the $j$-th bit of the encoding\footnote{This encoding is either $\bc(v_i)$ or $\phi_{v_i}$ depending on whether we are at Step~\ref{step:enc-ecc} or Step~\ref{step:enc-matrix} of Algorithm~\ref{Alg:overall-loc-rand}.} of its item $v_i$ with the sign of $b$.  If signs are equal, then the desired probability is $\frac{e^{\eps}}{1+e^{\eps}}$, otherwise it is $\frac{1}{1+e^{\eps}}$. Hence, the computational cost of this step (per user) is $O\left(T\log\left(m_{\sf PP}\right)+\log\left(m_{\sf FO}\right)\right)=O(\log\left(\log(d)\right)\log(1/\beta)+\log(n))$ where $m_{\sf PP}$ is the length of the encoding $\bc(v_i)$ used in the promise problem protocol $\prot^{\sf PP}$-$\sh_{\sf PP}$ and $m_{\sf FO}$ is the length of the encoding $\phi_{v_i}$ used in the frequency oracle protocol $\prot$-$\fo$. Thus, at worst the overall computational cost of the $1$-Bit protocol is the same as that of protocol $\prot$-$\sh$.

\begin{algorithm}[htb]
	\caption{$\Q^{\sf Full}_i$: $\eps$-Local Randomizer of User $i$ in $\prot$-$\sh$ (Algorithm\ref{Alg:SH-eff})}
	\begin{algorithmic}[1]
		\REQUIRE item $v_i\in\V$, privacy parameter $\eps$, seeds of $\hash$~ $s_1, ..., s_T$.
        \FOR{ $t=1$ to $T$}
		\FOR{ Channels $k=1$ to $K$ }
            {\STATE If $\hash(s_t, v_i)\neq k$, set $\bz^{(t,k)}_i=\R_i\left(\zeroB, \eps\right)$. Else, set $\bz^{(t,k)}_i=\R_i\left(\bc\left(v_i\right), \eps\right)$. \COMMENT{ $\bz^{(t,k)}_i$ denotes the report of user $i$ in the $k$-th channel in the $t$-th group.}}\label{step:enc-ecc}
        \ENDFOR
        \ENDFOR
   \STATE Set $\bz^{(fo)}_i=\R_i\left(\phi_{v_i}, \eps\right)$. \COMMENT{$\phi_{v_i}$ is the $v_i$-th column of $\Phi$ the encoding matrix in the construction of the frequency oracle $\fo$.}\label{step:enc-matrix}
\RETURN $\bz_i=\left(\bz^{(t,k)}_i,~ \bz^{(fo)}_i~:~ t=1, ..., T; k=1, ..., K\right)$.
\end{algorithmic}
	\label{Alg:overall-loc-rand}
\end{algorithm}

\fi


\ifnum\full=1
\section{Tight Lower Bound on the Error}
\label{sec:lower}

In this section, we provide a lower bound of $\Omega\left(\frac{1}{\eps}\sqrt{\frac{\log(d)}{n}}\right)$ on the error of frequency oracles and succinct histograms under the $(\eps, \delta)$-local privacy constraint. Our lower bound is the same for pure $\eps$ as for $(\eps, \delta)$-$\ldp$ algorithms when $\delta=o\left(\frac{1}{n\log(n)}\right)$. Namely, our lower bound shows that there is no advantage of $(\eps, \delta)$ algorithms over pure $\eps$ algorithms in terms of asymptotic error for all meaningful settings of $\delta$. In fact, it is standard to assume that $\delta =o(\frac{1}{n})$, say $\delta\approx\frac{1}{n^{\gamma}}, \gamma\geq 2$, since otherwise there are trivial examples of algorithms that are clearly non-private yet they satisfy the definition $(\eps, \delta)$ differential privacy (for example, see \cite[Example~2]{DKMMN06}).


Our lower bound matches the upper bound (for both frequency oracles and succinct histograms) discussed in previous sections. Hence, the efficient constructions given in Sections~\ref{subsec:JL}, \ref{sec:eff-protocol}, and \ref{subsec:eff-protocol-1bit} yield the optimal error. Our lower bound also shows that some previous constructions yield the optimal error, namely, the constructions of \cite{HKR10} and \cite{MS06}. However, as discussed in Section~\ref{sec:upper}, those constructions are computationally inefficient when used directly to construct succinct histograms.

\mypar{Our Technique} Our approach is inspired by the techniques used by \citet{DuchiJW13} to obtain lower bounds on the statistical minimax rate (expected worst-case error) of multinomial estimation in the pure $\eps$ local model. In a scenario where the item of each user is drawn independently from an unknown distribution on $\V$, we first derive a lower bound on the expected worst-case error (the minimax rate) in estimating the right distribution. We then show using standard concentration bounds that this implies a lower bound on the maximum error in estimating the actual frequencies of all the items in $\V$. To obtain a lower bound on the minimax rate, we first define the notion of an $\eta$-degrading channel which is a noise operator that, given a user's item as input, outputs the same item with probability $\eta$, and outputs a uniform random item from $\V$ otherwise. We compare two scenarios: in the first scenario, each user feeds its item first to an $\eta$-degrading channel, then feeds its output into its $(\eps, \delta)$ local randomizer to generate a report, whereas the second scenario is the normal scenario where the user feeds its item directly into its local randomizer. We then argue that a lower bound of $\Omega(1)$ on the minimax error in the first scenario implies a lower bound of $\Omega(\eta)$ in the second scenario. Next, we show that a lower bound of $\Omega(1)$ is true in the first scenario with an $\eta$-degrading channel when $\eta=\Omega\left(\frac{1}{\eps}\sqrt{\frac{\log(d)}{n}}\right)$, which gives us the desired lower bound. To derive the $\Omega(1)$ lower bound in the first scenario, we proceed as follows. First, we derive an upper bound of $O\left(\eps^2+\frac{\delta}{\eps}\log(d\eps/\delta)\right)$ on the mutual information between a uniform random item $V$ from $\V$ and the output of an $(\eps, \delta)$-local randomizer with input $V$. Then, we prove that the application of an $\eta$-degrading channel on a user's item amplifies privacy, namely, scales down both $\eps$ and $\delta$ by $\eta$. This implies that in the first scenario above with an $\eta$-degrading channel, the mutual information between a uniform item from $\V$ and the output of the $(\eps, \delta)$-local randomizer is $O\left(\eta^2\eps^2+\frac{\delta}{\eps}\log(d\eps/\delta)\right)$. We use such mutual information bound together with Fano's inequality to show that for $\eta=\Omega\left(\frac{1}{\eps}\sqrt{\frac{\log(d)}{n}}\right)$, the error in the first scenario is $\Omega(1)$.


\subsection{A Minimax Formulation}

\mypar{Notation and definitions} Let $\splex(d)\subset [0, 1]^d$ denote the probability simplex of $d$ corner points. Let $\mathcal{P}\in\splex(d)$ be some probability distribution over the item set $\mathcal{V}=[d]$\footnote{Without loss of generality, we will use $\V$ and $[d]$ interchangeably to denote the item set} . Users' items $v_i, i\in[n]$ are assumed to be drawn independently from the distribution $\Pc$.  For every $i\in[n]$, let $\Q_i:\mathcal{V}\rightarrow \mathcal{Z}$ be any $(\eps, \delta)$-local differentially private algorithm ($(\eps, \delta)$-local randomizer) used to generate a report $\bz_i\in\Z$ of user $i$ where $\mathcal{Z}$ is some arbitrary fixed set. All $\Q_i, i\in[n]$ use independent randomness.  Hence, all $\bz_i, i\in[n]$ are independent (but not necessarily identically distributed). Let $\A:\mathcal{Z}^n\rightarrow [0, 1]^d$ be an algorithm for estimating $\Pc$ based on the observations $\bz_1,..., \bz_n$. Let $\hP\in [0, 1]^d$ denote the output of $\A$. The expected $L_{\infty}$ estimation error for a given input distribution $\Pc$ and an estimation algorithm $\A$ is defined as
$$\Ec(\Pc; \A)\triangleq\E\left[\linf{\A(\bz_1,...,\bz_n)-\Pc}\right]=\E\left[\max\limits_{v\in[d]}\vert \hP_v-\Pc_v\vert\right]$$
where the expectation is taken over the distribution $\Pc$, the randomness of $\Q_i, i\in[n]$, and randomness (if any) of $\A$.

The minimax error (minimax rate) is defined as the minimum (over all estimators $\A$) of the maximum (over all distributions $\Pc\in\splex(d)$) error defined above. That is,
\begin{align}
\mmer&\triangleq \min\limits_{\A}\max\limits_{\Pc\in\splex(d)}\Ec(\Pc, \A).\label{minimax}
\end{align}
Let $\Bc: \mathcal{Z}^n\rightarrow [0,1]^d$ be an algorithm for estimating the frequency vector $\f=\frac{1}{n}\sum_{i=1}^n\be_{v_i}$ of the items in $\mathcal{V}$ where $\be_v$ is the $v$-th standard basis vector in $\re^d$. Let $\hbf\in[0, 1]^d$ denote the output of $\Bc$. The error incurred by $\Bc$ is given in Section~\ref{sec:intro}, that is,
\begin{align}
\err(\f; \Bc)=\linf{\hbf-\f}=\max\limits_{v\in[d]}\vert \hf(v)-f(v)\vert.\label{linf-err}
\end{align}

We first provide a lower bound of $\Omega\left(\frac{1}{\eps}\sqrt{\frac{\log(d)}{n}}\right)$ on (\ref{minimax}) and argue using Hoeffding's bound that this implies a lower bound of the same asymptotic order on the expectation of (\ref{linf-err}) which clearly proves our main result in this section.

\begin{lem}
For any $\eps=O(1)$ and $0\leq\delta\leq o\left(\frac{1}{n\log(n)}\right)$. For any sequence $\Q_i, i\in[n]$ of $(\eps, \delta)$-$\ldp$ algorithms, the minimax rate satisfies
$$\mmer=\Omega\left(\min\left(\frac{1}{\eps}\sqrt{\frac{\log(d)}{n}}, 1\right)\right)$$
where $\mmer$ is defined in (\ref{minimax}).
\end{lem}\label{lem:minmax-lower}

The above lemma is the central technical part of this section. The proof of this lemma is given separately in Section~\ref{sec:lower:sub-lemma-proof}. First, we formally state and prove our lower bound in the following section.

\subsection{Main Result}

\begin{thm}[Lower Bound on the Error of Private Histograms]
For any $\eps=O(1)$ and $0\leq\delta\leq o\left(\frac{1}{n\log(n)}\right)$. For any sequence $\Q_i:\mathcal{V}\rightarrow \mathcal{Z}, i\in[n]$ of $(\eps, \delta)$-$\ldp$ algorithms, and for any algorithm $\Bc: \mathcal{Z}^n\rightarrow [0,1]^d$, there exists a distribution $\Pc\in\splex(d)$ (from which $v_i, i\in[n]$ are sampled in i.i.d. fashion) such that the expected error with respect to such $\Pc$ is
$$\E\left[\err(\f; \Bc)\right]=\Omega\left(\min\left(\frac{1}{\eps}\sqrt{\frac{\log(d)}{n}}, 1\right)\right)$$
where $\err(\f; \Bc)$ is defined in (\ref{linf-err}).
\end{thm}\label{thm:lower-bound}

\begin{proof}
First, note that for the case where $n< \frac{\log(d)}{\eps^2}$, the above theorem follows directly from Lemma~\ref{lem:minmax-lower} since (as given in the proof of this lemma in Section~\ref{sec:lower:sub-lemma-proof}) our example distribution $\Pc$ will simply be a degenrate distribution and hence $\f=\Pc$ with probability 1.

Turning to the case where $n> \frac{\log(d)}{\eps^2}$, having Lemma~\ref{lem:minmax-lower} in hand, the proof of Theorem~\ref{thm:lower-bound} in this case becomes a simple application of Hoeffding's inequality.  First, fix $\eps,~\delta$ as in the theorem statement and let $\Q_i:\mathcal{V}\rightarrow \mathcal{Z}, i\in[n]$ of $(\eps, \delta)$-$\ldp$ mechanisms. Suppose, for the sake of contradiction, that there is an algorithm $\Bc:\mathcal{Z}^n\rightarrow[0, 1]^d$ such that for any distribution $\Pc\in\splex(d)$ from which the users' items $v_1,..., v_n$ are sampled in i.i.d. fashion, the error $\E\left[\err(\f; \Bc)\right]\neq\Omega\left(\frac{1}{\eps}\sqrt{\frac{\log(d)}{n}}\right)$. Let $\hbf$ denote the output of $\Bc$. Now, observe that for sufficiently large $n$ and $d$, we have
\begin{align}
\Ec(\Pc; \Bc)=\E\left[\linf{\hbf- \Pc}\right]&\leq \E\left[\linf{\hbf- \f}\right]+\E\left[\linf{\f- \Pc}\right]\leq \E\left[\err(\f; \Bc)\right]+\frac{\sqrt{3}}{2}\sqrt{\frac{\log(d)}{n}}\label{hoeff}
\end{align}
where the last inequality in (\ref{hoeff}) follows from using Hoeffding's inequality and the fact that
$$\E\left[\linf{\f- \Pc}\right]=\int\limits_{t=0}^{\infty}\Pr\left[\linf{\f- \Pc}\geq t\right] dt.$$
Hence, $\Ec(\Pc; \Bc)\neq\Omega\left(\sqrt{\frac{\log(d)}{\eps^2n}}\right)$. However, this contradicts Lemma~\ref{lem:minmax-lower}. Therefore, the proof is complete.
\end{proof}

\subsection{Proof of Lemma~\ref{lem:minmax-lower}}
\label{sec:lower:sub-lemma-proof}
We first introduce the notion of an $\eta$-degrading channel. For any $\eta\in[0,1]$, an $\eta$-degrading channel $\Weta:\V\rightarrow\V$ is a randomized mapping that is defined as follows: for every $v\in\V$,
\begin{align}
&\Weta(v)=\left\{\begin{array}{cc}
          v & ~\text{with probability }\eta \\
          U, & ~\text{with probability }1-\eta
        \end{array}\right.\label{def:eta-ch}
\end{align}
where $U$ is a uniform random variable over $\mathcal{V}$.

Let $\mmer_{\eta}$ be the minimax error resulting from the scenario where each user $i\in[n]$ with item $v_i\in\V$, first, apply $v_i$ to an independent copy $\Weta_i$ of an $\eta$-degrading channel, then apply the output to its $(\eps, \delta)$-$\ldp$ randomizer $\Q_i$ that outputs the report $\bz_i$. That is, $\mmer_{\eta}$ is the minimax error when $\Q_i(\cdot)$ is replaced with  $\Q_i\left(\Weta_i\left(\cdot\right)\right)$, $i\in[n]$. Our proof relies on the following lemma.
\begin{lem}
Let $\eta\in[0,1]$. If~ $\mmer\leq \frac{\eta}{10}$, then ~$\mmer_{\eta}\leq \frac{1}{10}$.
\end{lem}
\label{lem:contrapose}

\begin{proof}
Suppose that $\mmer\leq \frac{\eta}{10}$. Then, there is an algorithm $\A:\Z^n\rightarrow [0, 1]^d$ such that for any distribution $\tilde{\Pc}\in\splex(d)$ on $\V$, we have
\begin{align}
\E\left[\linf{\A\left(\Q_1(v_1), ..., \Q_n(v_n)\right)-\tilde{\Pc}}\right]&\leq\frac{\eta}{10}\label{proof-contra-minmax}
\end{align}
where $v_i$ is drawn from $\tilde{\Pc}$ independently for every $i\in[n]$.


Let $\We$ denote the distribution of the output of the $\eta$-degrading channel $\Weta$. Note that if the distribution of the input of $\Weta$ is $\Pc$, then
\begin{align}
\We&=\eta\Pc+(1-\eta)\U\label{out-dist-weta}
\end{align}
where $\U$ is the uniform distribution on $\V$.

Consider an algorithm $\A_{\eta}$ defined as follows. For any input $(\bz_1, ..., \bz_n)\in\Z^n$, $\A_{\eta}$ runs $\A$ on its input to obtain $\hP\in[0,1]^d$. Then, $\A_{\eta}$ computes a vector $\hP^{(\eta)}\in[0,1]^d$ whose entries $\hP_v^{(\eta)},~v\in[d],$ are given by $\frac{1}{\eta}\left(\hP_v - (1-\eta)\U\right)$ rounded to $[0, 1]$. Now, consider the scenario where we replace each $\Q_i(\cdot)$ with $\Q_i\left(\Weta\left(\cdot\right)\right)$ for all $i\in[n]$. Observe that for any distribution $\Pc$ of users' items, we have
\begin{align}
\Ec\left(\Pc, \A_{\eta}\right)&=\E\left[\linf{\hP^{(\eta)}-\Pc}\right]\leq\E\left[\linf{\frac{1}{\eta}\left(\hP-(1-\eta)\U\right)-\Pc}\right]=\frac{1}{\eta}\E\left[\linf{\hP-\We}\right]\label{proof-contra-1}\\
&=\frac{1}{\eta}\E\left[\linf{\A\left(\Q_1(y_1), ..., \Q_n(y_n)\right)-\We}\right]\leq\frac{1}{10}\label{proof-contra-2}
\end{align}
where $y_i$ is drawn from $\We$ independently for every $i\in[n]$. Note that the last equality in (\ref{proof-contra-1}) follows from (\ref{out-dist-weta}), and (\ref{proof-contra-2}) follows from (\ref{proof-contra-minmax}).
\end{proof}

Given Lemma~\ref{lem:contrapose}, our proof proceeds as follows. We show that for a setting of $\eta$ $=\Omega\left(\min\left(\sqrt{\frac{\log(d)}{\eps^2 n}}, 1\right)\right)$, we have $\mmer_{\eta}>\frac{1}{10}$ which, by Lemma~\ref{lem:contrapose}, implies that $\mmer=\Omega\left(\min\left(\sqrt{\frac{\log(d)}{\eps^2 n}}, 1\right)\right)$ which will complete our proof.

We consider the following scenario. Let $V$ be a uniform random variable on $\V$. Conditioned on $V=v$, for all $i\in[n]$, $v_i=v$, i.e., all users have the same item $v$ when $V=v$. Each user $i$ applies an independent copy of an $\eta$-degrading channel $\Weta_i$ to its item $v_i$. The output is then fed to the user's $(\eps, \delta)$-local randomizer $\Q_i$ that outputs the user's report $\bz_i$. Let $\G:\mathcal{Z}^n\rightarrow\V$ be an algorithm that, given the users' reports $\bz_1,..., \bz_n$, outputs an estimate $\hat{V}$ for the common item $V$.

\mypar{Fano's inequality} Let $\pre(\G)$ be the probability of error that $\G$ outputs a wrong hypothesis $\hat{V}\neq V$. That is,
$$\pre(\G)=\Pr\left[\G(\bz_1,..., \bz_n)\neq V\right].$$
Fano's inequality gives a lower bound on the probability of error incurred by \emph{any} such estimator $\G$:
\begin{align}
\mpre&\triangleq\min\limits_{\G}\pre(\G)\geq 1-\frac{I(v_1,...,v_n ; \bz_1,...\bz_n)+1}{\log(d)}\label{fano}
\end{align}
One can easily see that the minimax error $\mmer_{\eta}$ is bounded from below as
\begin{align}
\mmer_{\eta}&\geq\min\limits_{v\neq v'}\frac{\linf{\be_{v}-\be_{v'}}}{2}\mpre=\mpre\label{mmer-eta-fano}
\end{align}
To reach our goal, given (\ref{fano})-(\ref{mmer-eta-fano}), it suffices to show that for a setting of $\eta=\Omega\left(\min\left(\sqrt{\frac{\log(d)}{\eps^2 n}}, 1\right)\right)$, we have $\frac{I(v_1,...,v_n ; \bz_1,...\bz_n)}{\log(d)}\leq \frac{1}{2}$. This will be established using the following claims.

\begin{claim}
Let $V$ be uniformly distributed over $[d]$. Let $\Q:[d]\rightarrow\mathcal{Z}$ be an $(\eps, \delta)$-$\ldp$ algorithm and let $Z$ denote $\Q(V)$. Then, we have
$$I(V; Z)=O\left(\eps^2+\frac{\delta}{\eps}\log(d)+\frac{\delta}{\eps}\log(\eps/\delta)\right).$$
\end{claim}
\label{mut-inf-bd}

\begin{proof}
Let $M_v$ denote the probability density function of the output of $\Q(v)$, $v\in[d]$\footnote{To avoid cumbersome notation, we will ignore irrelevant technicalities of the measure defined on $\Z$ as well as the distinctions between probability mass functions and density functions and will use the notation $\int\limits_{z\in\Sc}M_v(z)dz$ to simply mean $\Pr[Z\in\Sc\vert V=v]$ whether $Z$ is discrete or continuous random variable.}. Let $\bM(z)=\frac{1}{d}\sum\limits_{v\in[d]}M_v(z)$, $z\in\Z$. For every $v\in [d]$, define
$$\bad^{(1)}_{v}\triangleq\left\{z \in \Z: ~ \frac{M_v(z)}{\bM(z)} > e^{2\eps}\right\}$$
and
$$\bad^{(2)}_{v}\triangleq\left\{z \in \Z: ~ \frac{M_v(z)}{\bM(z)} < e^{-2\eps}\right\}.$$
Let $\bad_v=\bad^{(1)}_{v} \cup \bad^{(2)}_{v}$ and $\bad=\cup_{v\in [d]}\left(\bad^{(1)}_{v}\cup\bad^{(2)}_{v}\right)$.  Let $B$ be a binary random variable that takes value $1$ whenever $Z\in\bad_V$ and $0$ otherwise. Now, observe that
\begin{align}
I(V; Z)&\leq I\left(V; Z, B\right)\leq I\left(V; Z |B \right) + H(B)\nonumber\\
&\leq I\left(V; Z \vert B = 0\right) + I\left(V; Z \vert B = 1\right)\Pr\left[\bad\right]+H(B)\label{main-bd-mut-inf}
\end{align}
where $H(B)$ denotes Shannon entropy of $B$. The above inequalities follow from the standard properties of the mutual information between any pair of random variables and the fact that $\Pr[B=1]=\Pr[\bad]$.

First, we consider the first term in (\ref{main-bd-mut-inf}). Conditioned on $B=0$, $Z$ lies in a set $\left\{z\in\Z:~ e^{-2\eps}\leq \frac{M_v(z)}{\bM(z)}\leq e^{2\eps}\right\}$. Hence, we can obtain a bound of $O(\eps^2)$ on $I\left(V; Z \vert B = 0\right)$ by applying techniques that were originally used for pure $\eps$ local differential privacy like those in \cite{DuchiJW13} (Corollary~1 therein).


Next, observe that $I\left(V; Z \vert B = 1\right)\leq H(V)\leq \log(d)$. Thus, it remains to bound $\Pr[\bad]$ (and consequently bound $H(B)$). Observe that
\begin{align}
\Pr\left[\bad^{(1)}_{v}\right]&=\int_{z\in\bad^{(1)}_{v}}M_v(z)dz>e^{2\eps}\int_{z\in\bad^{(1)}_{v}}\bM(z)dz\nonumber\\
&>e^{\eps}\int_{z\in\bad^{(1)}_{v}}M_v(z)dz-e^{\eps}\delta=e^{\eps}\Pr\left[\bad^{(1)}_{v}\right]-e^{\eps}\delta\nonumber
\end{align}
where the last inequality above follows from the fact that $\Q$ is $(\eps, \delta)$-differentially private. Hence, we get
$$\Pr\left[\bad^{(1)}_{v}\right]<\frac{e^{\eps}\delta}{e^{\eps}-1}=O(\delta/\eps).$$
Similarly, we can bound $\Pr\left[\bad^{(2)}_{v}\right]<\frac{e^{-\eps}\delta}{e^{\eps}-1}=O(\delta/\eps)$. Hence, $\Pr[\bad]=O(\delta/\eps)$ which gives us the required bound on the second term of \ref{main-bd-mut-inf}. Finally, note that the bound on $\Pr[\bad]$ implies a bound of $O\left(\frac{\delta}{\eps}\log(\eps/\delta)\right)$ on $H(B)$. This completes the proof.
\end{proof}

\begin{claim}[Privacy amplification via degrading channels]
The composition $\Q\left(\Weta\left(\cdot\right)\right)$ of an $\eta$-degrading channel $\Weta$ (defined in (\ref{def:eta-ch})) with an $(\eps, \delta)$-$\ldp$ algorithm $\Q$ yields a $\left(O(\eta\eps), O(\eta\delta)\right)$-$\ldp$ algorithm.
\end{claim}\label{priv-amp-by-eta-ch}

\begin{proof}


Fix any measurable subset $\Sc\subset\Z$. For any $v\in\V$, let $\Meta_v(\Sc)$ denote $\Pr\left[\Q\left(\Weta(v)\right)\in\Sc\right]$ and $\M_v(\Sc)$ denote $\Pr\left[\Q(v)\in\Sc\right]$. Fix any pair $v, v'\in\V$. Observe that
$$\Meta_v(\Sc)=\eta\M_v(\Sc)+(1-\eta)\frac{1}{d}\sum\limits_{u\in\V}\M_u(\Sc)$$
and
$$\Meta_{v'}(\Sc)=\eta\M_{v'}(\Sc)+(1-\eta)\frac{1}{d}\sum\limits_{u\in\V}\M_u(\Sc).$$
Hence, we can write $\Meta_{v'}(\Sc)$ as
\begin{align}
\Meta_{v'}(\Sc)&=\eta\left(\M_{v'}(\Sc)-\M_{v}(\Sc)\right)+\Meta_v(\Sc)\nonumber\\
&\leq \eta\left(e^{\eps}-1\right)\M_v(\Sc)+\Meta_v(\Sc)+\eta\delta\nonumber\\
&\leq \left(1+\eta e^{\eps}\left(e^{\eps}-1\right)\right)\Meta_v(\Sc)+e^{\eps}\eta\delta
\end{align}
The last inequality follows from the fact $\Q$ is $(\eps, \delta)$-differentially private, hence,
$$\M_v(\Sc)\leq e^{\eps}\E_{y\leftarrow \Weta(v)}\left[\M_y(\Sc)\right]+\delta=e^{\eps}\Meta_v(\Sc)+\delta.$$
We conclude the proof by noting that $1+\eta e^{\eps}\left(e^{\eps}-1\right)=e^{O(\eta\eps)}$ and $e^{\eps}\eta\delta=O(\eta\delta)$ since $\eps$ is $O(1)$.
\end{proof}

Putting these claims together with the fact that $I(v_1,...,v_n ; \bz_1,...\bz_n)\leq\sum_{i=1}^n I(v_i; \bz_i)$, we reach that, for $\delta=o\left(\frac{1}{n\log(n)}\right)$, we have
$$\frac{I(v_1,..., v_n ; \bz_1,...\bz_n)}{\log(d)}=O\left(\frac{n\eta^2\eps^2}{\log(d)}+\frac{1}{\log(n)}+\frac{1}{\log(d)}\right)$$
which, by an appropriate setting of $\eta=\Omega\left(\min\left(\sqrt{\frac{\log(d)}{\eps^2 n}}, 1\right)\right)$, can be made smaller than $1/2$. This completes the proof of Lemma~\ref{lem:minmax-lower}.


\else 
\section{Tight Lower Bound on Error}
\label{sec:lower}

We derive a matching lower bound on the error of $(\eps, \delta)$-$\ldp$ frequency oracles and succinct histograms for all $\delta=o\left(\frac{1}{n\log(n)}\right)$. Our lower bound implies that there is no advantage of $(\eps, \delta)$ algorithms over pure $\eps$ algorithms in terms of asymptotic error for all the meaningful settings of $\delta$. Our approach is inspired by some of the techniques in \cite{DuchiJW13} used for lower bounds on multinomial estimation error in the pure $\eps$ local model. We make their framework more modular, and  show that it can be used  to prove lower bounds for $(\eps,\delta)$-local differentially private protocols. Our lower bound is formally stated in the following theorem.

\begin{thm}
For any $\eps=O(1)$ and $0\leq\delta\leq o\left(\frac{1}{n\log(n)}\right)$. For any sequence $\Q_i:\mathcal{V}\rightarrow \mathcal{Z}, i\in[n]$ of $(\eps, \delta)$-$\ldp$ algorithms, and for any algorithm $\Bc: \mathcal{Z}^n\rightarrow [0,1]^d$, there exists a distribution $\Pc$ over $\V$ (from which $v_i, i\in[n]$ are sampled in i.i.d. fashion) such that the expected error with respect to such $\Pc$ satisfies
$$\E\left[\max\limits_{v\in[d]}\vert \hf(v)-f(v)\vert\right] =\Omega\left(\min\left(\frac{1}{\eps}\sqrt{\frac{\log(d)}{n}}, 1\right)\right).$$
\end{thm}\label{thm:lower-bound}

We discuss the main steps of our technique next. For more details on the proof, we refer the reader to the full version \cite{fulvBS15}.

\subsection{Our Technique}
In a scenario where the item of each user is drawn independently from an unknown distribution on $\V$, we first derive a lower bound on the expected worst-case error (the minimax rate) in estimating the right distribution. We then show using standard concentration bounds that this implies a lower bound on the maximum error in estimating the actual frequencies of all the items in $\V$. 

To obtain a lower bound on the minimax error, we first define the notion of an $\eta$-degrading channel which is a noise operator that, given a user's item as input, outputs the same item with probability $\eta$, and outputs a uniform random item from $\V$ otherwise. Formally, for any $\eta\in[0,1]$, an $\eta$-degrading channel $\Weta:\V\rightarrow\V$ is a randomized mapping that is defined as follows: for every $v\in\V$,
\begin{align}
&\Weta(v)=\left\{\begin{array}{cc}
          v & ~\text{with probability }\eta \\
          U, & ~\text{with probability }1-\eta
        \end{array}\right.\label{def:eta-ch}
\end{align}
where $U$ is a uniform random variable over $\mathcal{V}$.

We compare two scenarios: in the first scenario, each user feeds its item first to an $\eta$-degrading channel, then feeds its output into its $(\eps, \delta)$ local randomizer to generate a report, whereas the second scenario is the normal scenario where the user feeds its item directly into its local randomizer. We then show that a lower bound of $\Omega(1)$ on the minimax error in the first scenario implies a lower bound of $\Omega(\eta)$ in the second scenario. 

Thus, to reach our result, it would suffice to show that a lower bound of $\Omega(1)$ is true in the first scenario with an $\eta$-degrading channel when $\eta=\Omega\left(\frac{1}{\eps}\sqrt{\frac{\log(d)}{n}}\right)$. To derive the $\Omega(1)$ lower bound in the first scenario, we proceed as follows. 

First, we derive the following upper bound on the mutual information between a uniform random item $V$ from $\V$ and the output of an $(\eps, \delta)$-local randomizer with input $V$. 

\begin{claim}
Let $V$ be uniformly distributed over $\V$. Let $\Q:\V\rightarrow\mathcal{Z}$ be an $(\eps, \delta)$-$\ldp$ algorithm and let $Z$ denote $\Q(V)$. Then, we have
$$I(V; Z)=O\left(\eps^2+\frac{\delta}{\eps}\log(d)+\frac{\delta}{\eps}\log(\eps/\delta)\right).$$
\end{claim}
\label{mut-inf-bd}

Then, we prove that the application of an $\eta$-degrading channel on a user's item \emph{amplifies privacy}, namely, scales down both $\eps$ and $\delta$ by $\eta$. 

\begin{claim}
The composition $\Q\left(\Weta\left(\cdot\right)\right)$ of an $\eta$-degrading channel $\Weta$ (defined in (\ref{def:eta-ch})) with an $(\eps, \delta)$-$\ldp$ algorithm $\Q$ yields a $\left(O(\eta\eps), O(\eta\delta)\right)$-$\ldp$ algorithm.
\end{claim}\label{priv-amp-by-eta-ch}

This implies that in the first scenario above with an $\eta$-degrading channel, the mutual information between the users' items and the corresponding outputs of their $(\eps, \delta)$-local randomizers is
\begin{align}
&I(v_1,...,v_n ; \bz_1,...\bz_n)\leq\sum_{i=1}^n I(v_i; \bz_i)=O\left(n\eta^2\eps^2+n\frac{\delta}{\eps}\log(d\eps/\delta)\right)\nonumber
\end{align}
Fano's inequality implies that the minimax error in this scenario is bounded from below as
$$\mmer_{\eta}\geq 1-\frac{I(v_1,...,v_n ; \bz_1,...\bz_n)+1}{\log(d)}$$
Hence, by an appropriate setting of $\eta=\Omega\left(\min\left(\sqrt{\frac{\log(d)}{\eps^2 n}}, 1\right)\right)$ and for $\delta=o\left(\frac{1}{n\log(n)}\right)$, our mutual information bound above together with Fano's inequality implies that the error in the first scenario is $\Omega(1)$. This concludes the proof of our lower bound.

\fi



\section*{Acknowledgments}

R.B. and A.S. were funded by NSF awards \#0747294 and \#0941553. Some
of this work was done while A.S. was on sabbatical at Boston
University's Hariri Center for Computation and at Harvard University's
CRCS.
We
are grateful to helpful conversations with {\'U}lfar Erlingsson, Aleksandra Korolova, Frank McSherry, Ilya
Mironov, Kobbi
Nissim, and Vasyl Pihur.

\bibliographystyle{abbrvnat}
\bibliography{reference}

\end{document}

\fi

\usepackage{fullpage}
\usepackage{times}
\usepackage{graphicx,color}
\usepackage{array,float}
\usepackage{url}
\usepackage[usenames,dvipsnames]{xcolor}
\usepackage{amstext,amssymb,amsmath}
\usepackage{hyphenat}
\usepackage{amsthm}
\usepackage{verbatim}
\usepackage{bm}
\usepackage{paralist}
\usepackage{ulem}\normalem
\usepackage[numbers]{natbib}
\usepackage{todonotes}
\usepackage{paralist}
\usepackage{wrapfig}
\usepackage{hyperref}
\usepackage[noend]{algorithmic}
\usepackage{algorithm}

\newtheorem{lem}{Lemma}[section]
\newcommand{\thet}{\theta}
\newcommand{\linfty}[1]{\|#1\|_\infty}
\newtheorem{thm}[lem]{Theorem}
\newtheorem{cor}[lem]{Corollary}
\newtheorem{problem}[lem]{Problem}
\newtheorem{defn}[lem]{Definition}
\newtheorem{fact}[lem]{Fact}
\newtheorem{assumption}[lem]{Assumption}
\newtheorem{claim}[lem]{Claim}
\newcommand{\ip}[2]{\langle #1,#2\rangle}
\renewcommand{\algorithmicrequire}{\textbf{Input:}}
\renewcommand{\algorithmicensure}{\textbf{Output:}}
\newcommand{\nptheta}{\hat\theta}
\newcommand{\symdiff}{\Delta}
\newcommand{\privtheta}{\theta^{priv}}
\renewcommand{\paragraph}[1]{\vspace{3pt}\noindent\textbf{#1}}
\newcommand{\scrX}{\ensuremath{\mathcal{X}}}
\newcommand{\scrY}{\ensuremath{\mathcal{Y}}}
\newcommand{\scrZ}{\ensuremath{\mathcal{Z}}}		
\newcommand{\eL}{\mathcal{L}}
\newcommand{\beL}{\bar{\mathcal{L}}}
\newcommand{\bad}{\sf Bad}
\newcommand{\good}{\sf Good}
\newcommand{\hash}{\sf Hash}
\newcommand{\rA}{\ensuremath{\rightarrow}}
\newcommand{\rrA}{\ensuremath{\longrightarrow}}
\newcommand{\qB}{\ensuremath{\mathbf{q}}}
\newcommand{\XB}{\ensuremath{\mathbf{X}}}
\newcommand{\zeroB}{\ensuremath{\mathbf{0}}}
\newcommand{\sm}{\mbox{\textendash}}
\newcommand{\ltwo}[1]{\|#1\|_2}
\newcommand{\lone}[1]{\|#1\|_1}
\newcommand{\linf}[1]{\|#1\|_{\infty}}
\newcommand{\eps}{\epsilon}
\newcommand{\A}{\mathcal{A}}
\newcommand{\cost}{\mathsf{cost}}
\newcommand{\stat}{\mathsf{stat}}
\newcommand{\D}{\mathcal{D}}
\newcommand{\J}{J}
\newcommand{\G}{\mathcal{G}}
\newcommand{\hc}{\mathcal{H}}
\newcommand{\bx}{\mathbf{x}}
\newcommand{\by}{\mathbf{y}}
\newcommand{\bz}{\mathbf{z}}
\newcommand{\hbz}{\hat{\mathbf{z}}}
\newcommand{\bw}{\mathbf{w}}
\newcommand{\bolda}{\mathbf{a}}
\newcommand{\boldb}{\mathbf{b}}
\newcommand{\indx}{\mathsf{index}}
\newcommand{\bit}{\mathsf{bit}}
\newcommand{\bbz}{\bar{\mathbf{z}}}
\newcommand{\bM}{\bar{M}}
\newcommand{\zcp}{\mathcal{Z}_{+}}
\newcommand{\zcn}{\mathcal{Z}_{-}}
\newcommand{\hypsc}{\{-\frac{1}{\epsilon},~\frac{1}{\epsilon}\}^m}
\newcommand{\hypc}{\{-\frac{1}{\sqrt{m}},~\frac{1}{\sqrt{m}}\}^m}
\newcommand{\Pc}{\mathcal{P}}
\newcommand{\hP}{\hat{\mathcal{P}}}
\newcommand{\T}{\mathcal{T}}
\newcommand{\risk}{{\sf R}}
\newcommand{\vol}{{\sf Vol}}
\newcommand{\ind}{{\mathbf{1}}}
\newcommand{\mineig}{\mu}
\newcommand{\I}{\mathbb{I}}
\newcommand{\Ico}{\mathcal{I}_{\sf out}}
\newcommand{\Ici}{\mathcal{I}_{\sf in}}
\newcommand{\E}{\mathbb{E}}
\newcommand{\Sc}{\mathcal{S}}
\newcommand{\Ec}{\mathcal{E}}
\newcommand{\V}{\mathcal{V}}
\newcommand{\F}{\mathcal{F}}
\newcommand{\samp}{\mathsf{Samp}}
\newcommand{\empL}{\mathcal{L}}
\newcommand{\hatw}{\hat{w}}
\newcommand{\dist}{{\sf Dist}_{\infty}}
\newcommand{\hdist}{{\sf Dist}}
\newcommand{\htheta}{\widetilde\theta}
\newcommand{\ptheta}{\theta^\perp}
\newcommand{\dagw}{w^\dagger}
\newcommand{\tildew}{\tilde{w}}
\newcommand{\tildeF}{\tilde{F}}
\newcommand{\tildef}{\tilde{f}}
\newcommand{\re}{\mathbb{R}}
\newcommand{\B}{\mathbb{B}}
\newcommand{\Bc}{\mathcal{B}}
\newcommand{\enc}{{\sf Enc}}
\newcommand{\dec}{{\sf Dec}}
\newcommand{\He}{{\sf Heavhit}}
\newcommand{\coll}{{\sf Coll}}
\newcommand{\splex}{\mathsf{simplex}}
\newcommand{\Q}{\mathcal{Q}}
\renewcommand{\S}{\mathbb{S}}
\newcommand{\teps}{\tilde{\epsilon}}
\newcommand{\hw}{\hat{w}}
\newcommand{\hmu}{\hat{\mu}}
\newcommand{\hmuA}{\hat{\mu}_{\sf A}}
\newcommand{\muA}{\mu_{\sf A}}
\newcommand{\hmuC}{\hat{\mu}_{\C}}
\newcommand{\muC}{\mu_{\C}}
\newcommand{\hmug}{\hat{\mu}_{\good}}
\newcommand{\istr}{i^{\ast}}
\newcommand{\mU}{\mathcal{U}}
\newcommand{\grad}{\bigtriangledown}
\newcommand{\mypar}[1]{\smallskip
\noindent{\bf\em {#1}:}}
\newcommand{\boldpar}[1]{\smallskip
\noindent{\bf{#1}:}}
\newcommand{\etal}{\emph{et al.}}
\newcommand{\ldp}{\bf{LDP}}

\newcommand{\ignore}[1]{}
\newcommand{\tprivJ}{{{\tilde J}^{\text{priv}}}}
\newcommand{\tnonoiseJ}{{{J}^\#}}
\newcommand{\z}{z}
\renewcommand{\b}{b}
\newcommand{\TODO}[1]{{\bf TODO: #1}}
\newcommand{\NOTE}[1]{{\bf NOTE: #1}}
\newcommand{\Ced}{{\Delta_{\epsilon,\delta}}}
\newcommand{\name}{\textsc{GUPT}\xspace}
\newcommand{\nameplain}{GUPT\xspace}
\newcommand{\aname}{a \textsc{GUPT}\xspace}
\newcommand{\Aname}{A \textsc{GUPT}\xspace}
\newcommand{\Weta}{\mathbf{W}^{(\eta)}}
\newcommand{\We}{\mathcal{W}^{(\eta)}}
\newcommand{\W}{\mathcal{W}}
\newcommand{\M}{\mathcal{M}}
\newcommand{\Meta}{\mathcal{M}^{(\eta)}}
\newcommand{\U}{\mathcal{U}}
\newcommand{\R}{\mathcal{R}}
\newcommand{\Z}{\mathcal{Z}}
\newcommand{\be}{\mathbf{e}}
\newcommand{\f}{\mathbf{f}}
\newcommand{\bc}{\mathbf{c}}
\newcommand{\hf}{\hat{f}}
\newcommand{\hbf}{\hat{\mathbf{f}}}
\newcommand{\C}{\mathcal{C}}
\newcommand{\er}{\mathsf{Error}}
\newcommand{\poly}{\mathsf{poly}}
\newcommand{\btr}{\mathsf{Round}}
\newcommand{\genproj}{\mathsf{GenProj}}
\newcommand{\gen}{\mathsf{RndGen}}
\newcommand{\struct}{\mathsf{struct}}
\newcommand{\err}{\textsc{Err}}
\newcommand{\sh}{\mbox{S-Hist}}
\newcommand{\fo}{\mbox{FO}}
\newcommand{\prot}{{\sf PROT}}
\newcommand{\gprot}{{\sf GenPROT}}
\newcommand{\code}{{\sf code}(d, m, \zeta)}
\newcommand{\mmer}{\mathsf{MinMaxError}}
\newcommand{\Lap}{\mathsf{Lap}}
\newcommand{\List}{\textsc{List}}
\newcommand{\pre}{\mathsf{P_{error}}}
\newcommand{\mpre}{\mathsf{P_{min-error}}}

\newcommand{\paren}[1]{{\left({#1}\right)}}

\begin{document}

\date{}
\title{Local, Private, Efficient Protocols \\ for Succinct Histograms}
\author{Raef Bassily\thanks{Computer Science and Engineering
    Department, The Pennsylvania State
    University. \texttt{\{bassily,asmith\}@psu.edu}}
\and Adam Smith\footnotemark[1]\ \ \thanks{Work done while A.S. was on
sabbatical at Boston University and
Harvard University.}
}

\maketitle

\thispagestyle{empty}

\begin{abstract}
  We give efficient protocols and matching accuracy lower
  bounds for frequency estimation in the local model for differential
  privacy. In this model, individual users randomize their data
  themselves, sending differentially private reports to an untrusted
  server that aggregates them.

  We study protocols that produce a succinct histogram
  representation of the data. A succinct histogram is a list of the most
  frequent items in the data (often called ``heavy hitters'')  along
  with estimates of their frequencies; the frequency of all other
  items is implicitly estimated as 0.

  If there are $n$ users whose items come from a universe of size $d$,
  our protocols run in time polynomial in $n$ and $\log(d)$.  With
  high probability, they estimate the accuracy of every item up to
  error $O(\sqrt{\log(d)/(\eps^2n)})$.
  Moreover, we show that this much error is necessary, regardless of
  computational efficiency, and even for the simple setting where only
  one item appears with significant frequency in the data set.

  Previous protocols (Mishra and Sandler, 2006; Hsu, Khanna
  and Roth, 2012) for this task either ran in time $\Omega(d)$ or had
  much worse error (about $\sqrt[6]{\log(d)/(\eps^2n)}$), and the only
  known lower bound on error was $\Omega(1/\sqrt{n})$. \rbnote{I think the term ``efficient'' here does not fit.}

  We also adapt a result of McGregor et al (2010) to the local
  setting. In a model with public coins, we show that each user need
  only send 1 bit to the server. For all known local
  protocols (including ours), the transformation preserves
  computational efficiency.


\end{abstract}


\newpage
\thispagestyle{empty}

\tableofcontents

\newpage
\pagenumbering{arabic}

\addcontentsline{toc}{section}{Acknowledgments}
\section*{Acknowledgments}

R.B. and A.S. were funded by NSF awards \#0747294 and \#0941553. Some
of this work was done while A.S. was on sabbatical at Boston
University's Hariri Center for Computation and at Harvard University's
CRCS (supported by a Simons Investigator award to Salil Vadhan). We
are grateful to helpful conversations with {\'U}lfar Erlingsson,
Aleksandra Korolova, Frank McSherry, Ilya
Mironov, Kobbi
Nissim, and Vasyl Pihur. We thank Ilya Mironov for pointing out that
our 1-bit transformation follows the compression technique of
\citet{MMPRTV10}.

\small
\addcontentsline{toc}{section}{References}
\bibliographystyle{plainnat} \bibliography{reference}

\begin{thebibliography}{22}
\providecommand{\natexlab}[1]{#1}
\providecommand{\url}[1]{\texttt{#1}}
\expandafter\ifx\csname urlstyle\endcsname\relax
  \providecommand{\doi}[1]{doi: #1}\else
  \providecommand{\doi}{doi: \begingroup \urlstyle{rm}\Url}\fi

\bibitem[Agrawal and Haritsa(2005)]{AH05}
Shipra Agrawal and Jayant~R. Haritsa.
\newblock A framework for high-accuracy privacy-preserving mining.
\newblock In \emph{ICDE}, pages 193--204. IEEE Computer Society, 2005.

\bibitem[Blocki et~al.(2012)Blocki, Blum, Datta, and Sheffet]{BlockiBDS12jl}
Jeremiah Blocki, Avrim Blum, Anupam Datta, and Or~Sheffet.
\newblock The johnson-lindenstrauss transform itself preserves differential
  privacy.
\newblock In \emph{53rd Annual {IEEE} Symposium on Foundations of Computer
  Science, {FOCS} 2012, New Brunswick, NJ, USA, October 20-23, 2012}, pages
  410--419, 2012.

\bibitem[Duchi et~al.(2013)Duchi, Jordan, and Wainwright]{DuchiJW13}
John~C. Duchi, Michael~I. Jordan, and Martin~J. Wainwright.
\newblock Local privacy and statistical minimax rates.
\newblock In \emph{IEEE Symp. on Foundations of Computer Science (FOCS)}, 2013.

\bibitem[Dwork et~al.(2010{\natexlab{a}})Dwork, Noar, Pitassi, Rothblum, and
  Yekhanin]{panprivate}
C.~Dwork, M.~Noar, T.~Pitassi, G.~Rothblum, and S.~Yekhanin.
\newblock Pan-private streaming algorithms.
\newblock In \emph{Symposium on Innovations in Computer Science (ICS)},
  2010{\natexlab{a}}.

\bibitem[Dwork and Nissim(2004)]{DwNi04}
Cynthia Dwork and Kobbi Nissim.
\newblock Privacy-preserving datamining on vertically partitioned databases.
\newblock In \emph{CRYPTO}, LNCS, pages 528--544. Springer, 2004.

\bibitem[Dwork et~al.(2006{\natexlab{a}})Dwork, Kenthapadi, McSherry, Mironov,
  and Naor]{DKMMN06}
Cynthia Dwork, Krishnaram Kenthapadi, Frank McSherry, Ilya Mironov, and Moni
  Naor.
\newblock Our data, ourselves: Privacy via distributed noise generation.
\newblock In \emph{EUROCRYPT}, pages 486--503, 2006{\natexlab{a}}.

\bibitem[Dwork et~al.(2006{\natexlab{b}})Dwork, McSherry, Nissim, and
  Smith]{DMNS06}
Cynthia Dwork, Frank McSherry, Kobbi Nissim, and Adam Smith.
\newblock Calibrating noise to sensitivity in private data analysis.
\newblock In \emph{Theory of Cryptography Conference}, pages 265--284.
  Springer, 2006{\natexlab{b}}.

\bibitem[Dwork et~al.(2010{\natexlab{b}})Dwork, Naor, Pitassi, and
  Rothblum]{DNPR10}
Cynthia Dwork, Moni Naor, Toniann Pitassi, and Guy~N Rothblum.
\newblock Differential privacy under continual observation.
\newblock STOC '10, pages 715--724. ACM, 2010{\natexlab{b}}.

\bibitem[Erlingsson et~al.(2014)Erlingsson, Korolova, and Pihur]{Rappor14}
{\'{U}}lfar Erlingsson, Aleksandra Korolova, and Vasyl Pihur.
\newblock {RAPPOR:} randomized aggregatable privacy-preserving ordinal
  response.
\newblock \emph{ACM Symposium on Computer and Communications Security (CCS)},
  2014.

\bibitem[Evfimievski et~al.(2003)Evfimievski, Gehrke, and Srikant]{EGS03}
Alexandre Evfimievski, Johannes Gehrke, and Ramakrishnan Srikant.
\newblock Limiting privacy breaches in privacy preserving data mining.
\newblock In \emph{PODS}, pages 211--222. ACM, 2003.

\bibitem[Fanti et~al.(2015)Fanti, Pihur, and Erlingsson]{Rappor-unknowns15}
Giulia Fanti, Vasyl Pihur, and {\'U}lfar Erlingsson.
\newblock Building a rappor with the unknown: Privacy-preserving learning of
  associations and data dictionaries.
\newblock In \emph{arXiv:1503.01214 [cs.CR]}, 2015.

\bibitem[Gilbert et~al.(2002)Gilbert, Guha, Indyk, Muthukrishnan, and
  Strauss]{GGIMS02}
A.~C. Gilbert, S.~Guha, P.~Indyk, S.~Muthukrishnan, and M.~Strauss.
\newblock Near-optimal sparse fourier representations via sampling.
\newblock STOC 2002, pages 152--161. ACM, 2002.

\bibitem[Guruswami(2002)]{venkat-thesis}
Venkatesan Guruswami.
\newblock \emph{List Decoding of Error-Correcting Codes}.
\newblock PhD thesis, 2002.

\bibitem[Hsu et~al.(2012)Hsu, Khanna, and Roth]{HKR10}
Justin Hsu, Sanjeev Khanna, and Aaron Roth.
\newblock Distributed private heavy hitters.
\newblock In \emph{ICALP (1)}, pages 461--472, 2012.

\bibitem[Kasiviswanathan et~al.(2008)Kasiviswanathan, Lee, Nissim,
  Raskhodnikova, and Smith]{KLNRS08}
Shiva~Prasad Kasiviswanathan, Homin~K. Lee, Kobbi Nissim, Sofya Raskhodnikova,
  and Adam Smith.
\newblock What can we learn privately?
\newblock In \emph{FOCS}, 2008.

\bibitem[Kearns(1998)]{Kearns98}
Michael Kearns.
\newblock Efficient noise-tolerant learning from statistical queries.
\newblock \emph{Journal of the ACM}, 45\penalty0 (6):\penalty0 983--1006, 1998.
\newblock ISSN 0004-5411.
\newblock Preliminary version in {\emph{proceedings of STOC'93}}.

\bibitem[Kenthapadi et~al.(2012)Kenthapadi, Korolova, Mironov, and
  Mishra]{KenthapadiKMM12}
Krishnaram Kenthapadi, Aleksandra Korolova, Ilya Mironov, and Nina Mishra.
\newblock Privacy via the johnson-lindenstrauss transform.
\newblock \emph{CoRR}, abs/1204.2606, 2012.

\bibitem[Kushilevitz and Mansour(1991)]{Kushilevitz-Mansour91}
Eyal Kushilevitz and Yishay Mansour.
\newblock Learning decision trees using the fourier spectrum.
\newblock STOC'91, pages 455--464. ACM, 1991.

\bibitem[McDiarmid(1989)]{mcdiarmid}
Colin McDiarmid.
\newblock On the method of bounded differences.
\newblock In \emph{In Surveys in Combinatorics}, pages 148--188. Cambridge
  University Press, 1989.

\bibitem[McGregor et~al.(2010)McGregor, Mironov, Pitassi, Reingold, Talwar, and
  Vadhan]{MMPRTV10}
Andrew McGregor, Ilya Mironov, Toniann Pitassi, Omer Reingold, Kunal Talwar,
  and Salil~P. Vadhan.
\newblock The limits of two-party differential privacy.
\newblock In \emph{FOCS}, pages 81--90, 2010.

\bibitem[Mishra and Sandler(2006)]{MS06}
Nina Mishra and Mark Sandler.
\newblock Privacy via pseudorandom sketches.
\newblock In \emph{PODS}, pages 143--152. ACM, 2006.

\bibitem[Upadhyay(2013)]{Upadhyay13}
Jalaj Upadhyay.
\newblock Random projections, graph sparsification, and differential privacy.
\newblock In \emph{Advances in Cryptology - {ASIACRYPT}}, pages 276--295, 2013.

\end{thebibliography}

\end{document}